\documentclass[review]{elsarticle}

\usepackage{lineno,hyperref}
\journal{Journal of \LaTeX\ Templates} 

	\usepackage[utf8]{inputenc} 
	\usepackage[T1]{fontenc}
	\usepackage{	amsmath,%
								amssymb,%
								amsthm,%
								amsfonts,%
								mathtools,%
								xfrac,%
								bm,%
								mathrsfs,%
								commath}
	\usepackage{url}
	\usepackage{graphicx}
	\usepackage{enumitem}
	\usepackage{	lineno,%
								hyperref}
	\usepackage{ifthen}
	\usepackage{datetime}
	\usepackage{epsfig}
	\usepackage{pict2e}
	\usepackage{color}
	\usepackage{csquotes}
	\usepackage{dsfont}
	\usepackage{multicol}
	\usepackage{accents}
	\usepackage{tikz}
		\usetikzlibrary{fit}
		\usetikzlibrary{arrows.meta}
		\usetikzlibrary{bending}
		\usetikzlibrary{positioning}
		\tikzset{myarrow/.style={arrows={-Latex[round,bend]}}}
	\usepackage[]{natbib}

	\newtheorem{theorem}{Theorem}
	\newtheorem{corollary}{Corollary}
	\newtheorem{lemma}{Lemma}
	
	\newtheorem{definition}{Definition}
	\newtheorem{remark}{Remark}

	\newcommand{\describes}{\rightrightarrows}
	\newcommand{\bbset}[1]{\mathbb{#1}}
	\newcommand{\N}[0]{\bbset{N}}      			
	\newcommand{\Nplus}{\N_{+}}
	\newcommand{\Z}[0]{\bbset{Z}}      			
	\newcommand{\Q}[0]{\bbset{Q}}      			
	\newcommand{\R}[0]{\bbset{R}}      			
	\newcommand{\C}[0]{\bbset{C}}      			
	\newcommand{\Rc}[0]{\bbset{R}_c}   			
	\newcommand{\Cc}[0]{\bbset{C}_c}      	
	\newcommand{\Cs}[0]{\mathcal{C}}									
	\newcommand{\Ca}[0]{\mathcal{C}^*}								
	\newcommand{\Cas}[1]{\mathrlap{\phantom{\mathcal{C}}^{*}}\mathcal{C}_{#1}}		
	\newcommand{\Bs}[0]{\mathcal{B}}      						
	\newcommand{\Es}[0]{\mathcal{E}}      						
	\newcommand{\CBs}[0]{\mathcal{C}\mathcal{B}}      			
	\newcommand{\CEs}[0]{\mathcal{C}\mathcal{E}}      			
	\newcommand{\WEs}[0]{\mathcal{W}\mathcal{E}}				
	\newcommand{\SG}{\mathcal{G}}								
	\newcommand{\Oa}{\mathcal{O}^{\approx}} 
	\newcommand{\Ob}{\mathrlap{\phantom{\mathcal{O}}_{\sigma}}\mathcal{O}^{<}}   
	\newcommand{\As}[0]{\mathcal{A}}
	\newcommand{\Ms}[0]{\mathcal{M}}
	\newcommand{\Ns}[0]{\mathcal{N}}							
	\newcommand{\Qs}[0]{\mathcal{Q}}	
	\newcommand{\fa}{f^{\mathrm{a}}}	%
	\newcommand{\fb}{f^{\mathrm{b}}}	%
	\newcommand{\fc}{f^{\mathrm{c}}}	%
	\newcommand{\Fta}{\Ft{f}^{\mathrm{a}}}	%
	\newcommand{\Ftb}{\Ft{f}^{\mathrm{b}}}	%
	\newcommand{\Ftc}{\Ft{f}^{\mathrm{c}}}	%
	\newcommand{\Gop}{\mathscr{G}}			
	\newcommand{\ind}[0]{\mathds{1}}		
	\newcommand{\tm}{{\bm{m}}}			
	\newcommand{\ts}{{\bm{s}}}			
	\renewcommand{\th}{{\bm{h}}}			
	\newcommand{\rp}{\mathrlap{\phantom{r}^{\hspace{1pt}\prime}}r} 					
	\newcommand{\rpp}{\mathrlap{\phantom{r}^{\hspace{1pt}\prime\prime}}r} 		
	\newcommand{\bp}{\mathrlap{\phantom{b}^{\hspace{1pt}\prime}}b}					
	\newcommand{\ap}{\mathrlap{\phantom{a}^{\hspace{1pt}\prime}}a}					
	\newcommand{\di}[1]{\;\mathrm{d}#1} 	
	\newcommand{\Ft}[1]{\hat{#1}}         
	\newcommand{\iu}[0]{i}                
	\newcommand{\spacedot}[0]{\,\cdot\,}      
	\newcommand{\TM}[0]{\text{TM}}
	\newcommand{\OM}[0]{\text{OM}}
	\newcommand{\TMu}[0]{\overline{\TM}_{\text{BW}}}
	\newcommand{\TMl}[0]{\underline{\TM}_{\text{BW}}}
	\newcommand{\bw}{B}														
	\newcommand{\qnt}{\mathsf{Q}}									
	\newcommand{\en}{\varphi}											
	\DeclareMathOperator{\sinc}{sinc}            	
	\def\e{\mathop{\mathrm{e}}\nolimits}					
	\newcommand{\tot}{\rightarrow}								
	\renewcommand{\part}{\hookrightarrow}						
	\newcommand{\At}[0]{A_{\mathrm{T}}^{\Phi}}		
	\DeclareMathOperator*{\mima}{\Theta\vphantom{p}} 
	\DeclareMathOperator*{\suppess}{ess\hspace{1pt}supp} 
	\let\inf\relax\DeclareMathOperator*{\inf}{inf\vphantom{p}}	
\vfuzz \hfuzz
\newboolean{arxiv}

\begin{document}
	\setboolean{arxiv}{true} 
	
	\begin{frontmatter}
	\title{On the Arithmetic Complexity of the Bandwidth of Bandlimited Signals}

\author{Holger Boche\fnref{fna}}
\fntext[fna]{Holger Boche is with the Technische Universit\"at M\"unchen,
							Lehrstuhl f\"ur Theoretische Informationstechnik, 80290 Munich, Germany, and the Munich
							Center for Quantum Science and Technology (MCQST), Schellingstr. 4, 80799 
							Munich, Germany. e-mail: boche@tum.de.}
\author{Yannik N. Böck\fnref{fnb}}
\fntext[fnb]{Yannik~N.~B{\"o}ck is with the Technische Universit\"at M\"unchen,
							Lehrstuhl f\"ur Theoretische Informationstechnik, 80290 Munich, Germany. e-mail: yannik.boeck@tum.de.}
\author{Ullrich J. Mönich\fnref{fnb}}
\fntext[fnc]{Ullrich~J.~M{\"o}nich is with the Technische Universit\"at M\"unchen,
							Lehrstuhl f\"ur Theoretische Informationstechnik, 80290 Munich, Germany. e-mail: moenich@tum.de.}

	\begin{abstract}
		\ifthenelse{\boolean{arxiv}}{\normalsize}{}
		The bandwidth of a signal is an important physical property that is of relevance in many
		signal- and information-theoretic applications. In this paper we study questions related to the computability of the bandwidth of computable
		bandlimited signals. To this end we employ the concept of Turing computability, which exactly describes what
		is theoretically feasible and can be computed on a digital computer.
		Recently, it has been shown that there exist computable bandlimited signals with finite energy, the actual bandwidth
		of which is not a computable number, and hence cannot be computed on a digital computer.
		In this work, we consider the most general class of band-limited signals, together with different computable representations thereof. 
		Among other things, our analysis includes a characterization of the arithmetic complexity 
		of the bandwidth of such signals and
		yields a negative answer to the question of whether it is at least possible to compute non-trivial upper or lower bounds for the 
		bandwidth of a bandlimited signal. 
		Furthermore, we relate the problem of bandwidth computation to the theory of oracle machines.
		In particular, we consider halting and totality oracles, which belong to the most frequently 
		investigated oracle machines in the theory of computation.
	\end{abstract}

	\ifthenelse{\boolean{arxiv}}
		{	\end{frontmatter}\pagebreak%
		}%
		{	
		}%

\section{Introduction}\label{sec:introduction}
	 
	\ifthenelse{\boolean{arxiv}}
		{	The~%
		}%
		{	\IEEEPARstart{T}{he}~%
		}%
	applications of bandlimited signals are a prominent field of research within the
	community of information theory. While most real-world physical systems are analog and 
	continuous in time, the actual processing of information is often done on digital devices that operate in discrete-time
	computational cycles. Hence, the conversion of signals from the analog to the 
	digital domain and vice versa is indispensable for modern technology \cite{isermann89_book, ogata95_book}.
	The link between both domains is established by various sampling theorems
	\cite{shannon49,higgins96_book}, including signal recovery in the presence of noise \cite{pawlak96}, 
	estimates for the error arising from finite-length sampling-series approximations \cite{cambanis82a},
	and refined theories for the treatment of multi-band signals \cite{landau67a,herley99}. 
	All of the mentioned applications include the bandwidth of the involved signals as an essential parameter.
	The processing of discrete-time representations of bandlimited signals has been studied as well, see e.g.
	\cite{habib01,monich17b}. There, the authors consider the replication of time-continuous LTI systems
	in the discrete-time domain. Again, the bandwidth of the involved signals appears as a crucial quantity.  
	Furthermore, bandlimited signals play a significant role in wireless communication systems, where the
	spectrum of the transmit signal has to be controlled in order to not interfere with other
	systems \cite{weinstein71,bingham90}.

	For a bandlimited signal \(f\), we refer to the smallest number
	\(\sigma\), such that \(f\) is bandlimited with bandwidth \(\sigma\), as the the actual bandwidth of \(\bw(f)\). 
	According to Shannon's sampling theorem, a bandlimited signal \(f\) with finite energy is uniquely
	determined by the sequence of samples \((f(\sfrac{k}{r}))_{k \in \Z}\) if  
	\(r \geq r_{\text{min}} = \bw(f)/\pi\) holds true, in which case it may be reconstructed
	by means of the Shannon sampling series. Hence, the actual bandwidth of a bandlimited signal
	is a relevant quantity.

	In this paper we will study questions related to the computability of the actual bandwidth
	\(\bw(f)\) of computable bandlimited signals.
	Our analysis is based on the theory of Turing computability, which
	characterizes the fundamental limits of digital computation.
	There exists a variety of problems that have been shown to be uncomputable on a digital computer, e.g.,
	the computation of the Fourier transform for certain signals \cite{boche19g,boche19f} or
	the spectral factorization \cite{boche19h}. That is, all of these problems
	lack a way to control the approximation error involved in the computation.

	\begin{figure}
		\centering
		\begin{tikzpicture}[font=\small]
			\node[rectangle,draw,inner ysep=1em,inner xsep=0.5em,align=center] (n1) {Algorithm\\\(\TM\)};
			\node[left=2em of n1,align=left] (n-f) {Input\\data \(f\)};
			\node[above=1.7em of n1] (n-eps) {Desired approximation error \(\epsilon\)};
			\node[right=1.9em of n1,text width=4.9cm] (n-res) {Output: \(\TM(f,\epsilon)\) such
				that\\ \(\lVert \text{"true solution"} - \TM(f,\epsilon) \rVert < \epsilon\)};
			\draw[myarrow] (n-f)--(n1);
			\draw[myarrow] (n-eps)--(n1);
			\draw[myarrow] (n1)--(n-res);
		\end{tikzpicture}
		\caption{%
			We adopt the computability model considered in \cite{boche21d_accepted}: The algorithm \(\TM\) gets two inputs: the data \(f\) and the desired
			approximation error \(\epsilon\). The computed output \(\TM(f,\epsilon)\) is guaranteed to
			be \(\epsilon\)-close to the true solution.}
		\label{fig:eps-error-input}
	\end{figure}
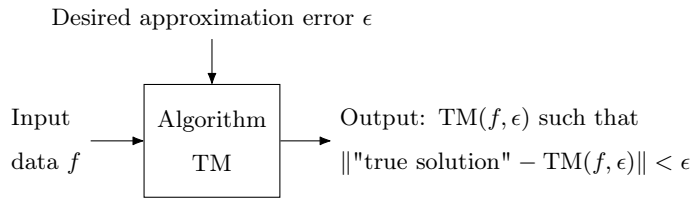

	Today, simulations in science and engineering rarely treat the involved approximation error explicitly, 
	which, for example, can be observed in the absence of error bars in plots.
	In those cases, the computer computes some rational number, which serves as an approximation 
	of the solution, without any quality guarantees.
	A framework which establishes such quality guarantees, i.e., an algorithmic control of the approximation error,
	is provided by the concept of computability. There, the algorithm obtains the desired approximation error
	\(\epsilon\), which could be the maximum tolerable error, along with the actual data \(f\) as an input, and  
	returns a solution that satisfies the error specification. 
	This algorithmic control of the error is illustrated in Fig.~\ref{fig:eps-error-input}.

	In \cite{boche21d_accepted}, it has been shown that there exist computable bandlimited signals \(f\)
	for which the actual bandwidth \(\bw(f)\) is not a computable number. 
	This means that there cannot exist any algorithm for the computation of \(\bw(f)\) with an
	effective control of the approximation error. 
	To the best of our knowledge, \cite{boche21d_accepted} is the first work on computable signals that
	satisfy essential properties of high practical relevance, like finite energy, continuity, etc., and 
	at the same time defy computability with respect to a fundamental signal parameter.
	
	In this work we aim to further develop the theory of computability of bandlimited signals. 
	Several open questions and conjectures were posed in \cite{boche21d_accepted} with regard to the class of computable 
	entire functions of exponential type \(\CEs_\pi\), which corresponds to the most general class of bandlimited signals:
	\begin{enumerate}
		\item For specific classes of computable bandlimited signals, the bandwidth \(\bw(f)\) is the
			limit value of a monotonically non-decreasing computable sequence of rational numbers.
			Does this hold true for the class \(\CEs_\pi\) as well? It was conjectured that the answer is no. 
		\item For specific classes of computable bandlimited signals, the subset of signals that satisfy
			\(\bw(f) > \lambda\) for a given \(\lambda\) is semi-decidable.
			Does this hold true for the class \(\CEs_\pi\) as well? Again, it was conjectured that the answer is no. 
		\item Even if it not possible to compute \(\bw(f)\) for problematic signals in \(\CEs_\pi\), it might still
			be possible to compute meaningful bounds for \(\bw(f)\). Hence, the question is: Do there exist Turing machines \(\TMl\)
			and \(\TMu\) such that for all signals \(f\) in \(\CEs_\pi\) that satisfy \(\bw(f) \leq \pi\), we have
			\begin{align*} \TMl(f) \leq \bw(f) \leq \TMu(f) ?
			\end{align*}
			It was conjectured that the only Turing machines that satisfy this requirement yield trivial values,
			i.e., \(\TMl(f) = 0\) and \(\TMu(f) = \pi\), for all signals in \(\CEs_\pi\).
	\end{enumerate}
	In this work, we provide a comprehensive study of the computability properties of the bandwidth 
	\(\bw(f)\) of signals \(f\) in the class \(\CEs_\pi\), proving all of the above conjectures correct. 
	As indicated above, the class \(\CEs_\pi\) respects the common general
	definition of bandlimited signals, which is a broad extension of the class of admissible signals compared
	to \cite{boche21d_accepted}. In particular, we present a sharp characterization of the \enquote{degree} 
	of non-computability of the number \(\bw(f)\) in the following sense. On the one hand, the arithmetic complexity
	of the number \(\bw(f)\) can never exceed the class \(\Pi_2\) (which will subsequently be introduced in a formal manner).
	On the other hand, for every number \(x \in \Pi_2 \cap [0,\pi]\), there exists a signal \(f\in\CEs_\pi\) with
	\(\bw(f)\) equal to \(x\). This characterization is analogous to the notion of 
	achievability and converse in information theory.

	Our analysis leads to interesting insights about the general limits of
	computability. In recent years, attempts have been made to push the boundaries of computing
	through various approaches, including analog, neuromorphic and quantum hardware.
	At least in theory, the computational capabilities of these technologies go beyond those of digital
	computers. For the question of whether any of them may also yield practical advantages at some point in the future, 
	it is essential to precisely understand the phenomena underlying the limitations of digital computing, especially in the
	context of mathematical models that describe practically relevant problems from engineering and 
	science. 
	
	Furthermore, our analysis uncovers interesting property of bandlimited signals, concerning the relation between
	the arithmetic complexity of \(\bw(f)\) and the structural properties of \(f\) in the time-domain. To the authors' 
	knowledge, no distinguished physical quantity other than
	the bandwidth of bandlimited signals has yet been identified to show a comparable behavior.
	The class \(\CEs_\pi\) includes the subclass \(\CBs_\pi^\infty\) of bandlimited signals 
	with finite \(L^\infty\)-norm, which, in contrast to \(\CEs_\pi\), exhibits a Banach space structure. In particular, this structure restricts the
	decay of the corresponding signals on the time-axis. While all signals in \(\CEs_\pi\) 
	are computationally well-behaved in the time domain, we observe a jump in the arithmetic complexity 
	of the bandwidth when extending the class of feasible signals from \(\CBs_\pi^\infty\) to \(\CEs_\pi\).
	Hence, the computability properties of the bandwidth are coupled directly with the presence or absence 
	of a Banach space structure in the time domain. Although not previously observed, the authors believe that 
	such characteristics may occur for a variety of mathematical models used in the
	applied sciences.
	
	The remainder of the paper is structured as follows. Sections \ref{sec:BandlimitedSignals} to \ref{sec:hierarchy_real_numbers}
	are dedicated to preliminaries. There, we introduce the basic concepts of bandlimited signals, Turing machines and 
	\(A\)-computable functions, as well as the arithmetical hierarchy of real numbers, which provides the theoretical framework
	for classifying different degrees of uncomputability. Our main results are presented in Sections 
	\ref{sec:ComputableBandlimitedSignals} to \ref{sec:OracleComputations}. We consider 
	different computable representations of bandlimited signals and derive the answers to the above-mentioned conjectures. 	
	Furthermore, we characterize the problem of computing the bandwidth of bandlimited signals by means of the 
	arithmetical hierarchy of real numbers. Last but not least, we relate the problem of computing the 
	bandwidth of bandlimited signals to a class of oracle computation machines. The paper closes in Section 
	\ref{sec:relation-prior-work}, with a discussion on our results and their implications.

\section{Bandlimited Signals}\label{sec:BandlimitedSignals}
	Our analysis of bandlimited signals is based on commonl used definitions,
	several of which have already been employed in \cite{boche21d_accepted}. For the sake of self-containedness,
	we introduce all of the relevant definitions in the following.
	
	By \(\Nplus := \{1,2,3,\ldots\}\), we denote the set of (positive) natural numbers.
	By \(\N := \Nplus \cup \{0\}\) the set of natural numbers including zero. 
	For \(\Omega \subseteq \R\), let \(L^p(\Omega)\), \(1\leq p < \infty\), be the space of all
	measurable, \(p\)-th power Lebesgue integrable functions on \(\Omega\), with the usual norm
	\(\lVert\spacedot\rVert_p\).
	A function \(f\) is said to be entire if it is defined and holomorphic on all of \(\C\).
	\begin{definition}\label{de:bandlimited} 
		An entire function \(f\) is called \emph{bandlimited} if
		there exists \(\sigma \geq 0\) such that for all \(\epsilon > 0\) there exists a constant
		\(C(\epsilon)\) with
		\begin{align}\label{eq:exponential_type}
			\lvert f(z) \rvert
			\leq
			C(\epsilon) \e^{  (\sigma+\epsilon) \lvert z \rvert }
		\end{align}
		for all \(z\in \C\) \cite{levin96_book,higgins96_book}.
		By \(\Es_{\sigma}\) we denote the set of all entire functions that are bandlimited with
		bandwidth \(\sigma\).
	\end{definition}
	
	In particular, we will consider signals that are bandlimited with bandwidth \(\pi\).

	According to the definition above, \(f \in \Es_{\sigma_1}\) implies
	\(f \in \Es_{\sigma_2}\) whenever \(\sigma_1\) satisfies \(\sigma_1 \leq \sigma_2\). 
	That is, a signal that is bandlimited with bandwidth \(\sigma_1\) is also bandlimited with
	any bandwidth \(\sigma_2\) larger than \(\sigma_1\). 
	For a given bandlimited signal \(f\), we denote by
	\begin{align}\label{eq:def_bandwidth}
		\bw(f)
		=
		\min \{ \sigma \geq 0 \colon f \in \Es_{\sigma}\}
	\end{align}
	the \emph{actual} bandwidth of the signal.
	
	\begin{remark}
		For an entire function \(f\) that satisfies \eqref{eq:exponential_type} for some \(\sigma\in\R\),
		the minimum in \eqref{eq:def_bandwidth} does exist. For details, see \cite[Appendix~B, p.~15]{boche21d_accepted}. 
	\end{remark}

	In the following, we introduce further signal spaces with practically relevant
	properties. If we additionally restrict the set \(\Es_\pi\) to signals with finite \(L^1\)-norm in the time domain, i.e.,
	\begin{align*}
		\int_{-\infty}^{\infty} \lvert f(t) \rvert \di{t} < \infty ,
	\end{align*}
	we obtain the Bernstein space \(\Bs_{\pi}^{1}\).
	On the other hand, if we restrict the set \(\Es_\pi\) to signals with well-defined Fourier transform in
	\(L^2\), we obtain the Bernstein space \(\Bs_{\pi}^{2}\).
	According to Plancherel's theorem, these signals also have a finite \(L^2\)-norm in the time
	domain. The Bernstein space \(\Bs_{\pi}^{2}\) is the frequently used space of bandlimited signals with finite energy.
	According to the Paley--Wiener theorem \cite[Theorem~7.2, p.~68]{higgins96_book}, the
	support of the Fourier transform \(\Ft{f}\) of a signal \(f \in \Bs_{\pi}^2\) is contained
	in \([-\pi,\pi]\), and we have
	\begin{align*}
		f(t)
		=
		\frac{1}{2\pi} \int_{-\pi}^{\pi} \Ft{f}(\omega) \e^{\iu \omega t} \di{\omega} .
	\end{align*}
	Hence, for the space \(\Bs_{\pi}^2\) we have a further, different characterization of the actual
	bandwidth.
	For \(f \in \Bs_{\pi}^2\), \(\bw(f)\) is the smallest number \(\sigma>0\) such that
	\begin{align*}
		f(t)
		=
		\frac{1}{2\pi} \int_{-\sigma}^{\sigma} \Ft{f}(\omega) \e^{\iu \omega t} \di{\omega}
	\end{align*}
	holds true for all \(t \in \R\).
	According to Plancherel's identity, this is also the smallest \(\sigma > 0\) such that
	\begin{align*}
		\int_{-\infty}^{\infty} \lvert f(t) \rvert^2 \di{t}
		=
		\frac{1}{2\pi} \int_{-\sigma}^{\sigma} \lvert \Ft{f}(\omega) \rvert^2 \di{\omega} .
	\end{align*}
	is satisfied.
	The actual bandwidth \(\bw(f)\) of a bandlimited signal \(f\) is a distinguished quantity,
	because it determines the minimum sampling rate that is required so that the samples
	uniquely determine \(f\).

	The general definition of the Bernstein spaces is as follows.
	\begin{definition}
		The Bernstein space \(\Bs_{\pi}^p\), \(1 \leq p \leq \infty\), consists of all
		functions in \(\Es_{\pi}\), whose restriction to the real line is in \(L^{p}(\R)\)
		\cite[p.~49]{higgins96_book}.
		The norm for \(\Bs_{\pi}^p\) is given by the \(L^p\)-norm on the real line.
	\end{definition}
	\begin{remark}
		We have \(\Bs_{\pi}^r \subsetneq \Bs_{\pi}^s \subsetneq \Es_{\pi}\)
		for all \(1 \leq r < s \leq \infty\). 
	\end{remark}

	We will discuss properties of signals \(f \in \Es_{\pi}\) next. 
	Since every signal \(f \in \Es_{\pi}\) is entire, it can be represented as a power series
	that converges uniformly on all compact subsets of \(\C\). Denote the \(n\)-th derivative of \(f\) by \(f^{(n)}\).
	Then, defining\linebreak \(a_n := f^{(n)}(0)\) for all \(n\in\N\), we have
	\begin{align}\label{eq:power-series}
		f(z)
		=
		\sum_{n=0}^{\infty} \frac{a_n}{n!} z^n 
		=
		\sum_{n=0}^{\infty} \frac{f^{(n)}(0)}{n!} z^n,	\quad z \in \C.
	\end{align}
	The actual bandwidth \(\bw(f)\) can be determined directly from the family of coefficients \((a_n)_{n\in\N}\)
	according to
	\begin{align}\label{eq:actual-bandwidth-limsup-sqrt-an}
		\bw(f)
		=
		\limsup_{n \to \infty} \sqrt[n]{\lvert a_n \rvert} .
	\end{align}
	For details, see \cite[pp.~356]{doetsch50_book} or \cite[Theorem~3, p.~6]{levin96_book}.

\section{Preliminaries on Turing Machines and Recursive Functions}
\label{sec:prel_on_TM}

	The theory of Turing machines, recursive functions and computable analysis are well-established fields in theoretical computer science. 
	Nevertheless, in order to establish a self-contained work, we introduce 
	all definitions and results that will be required subsequently, even if they have already been given in \cite{boche21d_accepted}.
	A comprehensive treatment of the topic may be found in \cite{weihrauch00_book,pour-el89_book,boolos02_book,avigad14,soare87_book,Soare2016}.

	Turing machines, as introduced by \citeauthor{turing36} in \cite{turing36,turing37}, are a mathematical model of what we intuitively understand as
	computation machines. In this sense, they yield an abstract idealization
	of today's real-world computers. Even though the model is relatively simple in structure, 
	any algorithm that can be executed by a real-world computer can be simulated by a Turing machine. 
	In contrast to real-world computers, however, Turing machines
	are not subject to any restrictions regarding energy consumption, computation
	time or memory size. All computation steps on a Turing machine are furthermore
	assumed to be executed with zero chance of error. Thus, computability in the sense of Turing is the exact characterization of
	what can be achieved by digital hardware, e.g., central processing units
	(CPUs), digital signal processors (DSPs), or field programmable gate arrays (FPGAs), if
	practical limitations, such as energy constraints, computing errors, and hardware
	restrictions, are disregarded. %

	In formal terms, a Turing machine consists of a formal language over a finite alphabet, together with a list of 
	transformation rules for the associated words. The transformation rules can be seen as an \enquote{algorithm}, 
	where the words represent the \enquote{data} being processed. Since formal languages exhibit a number of intuitive encodings into the set of natural numbers
	(cf. Remark \ref{rem:goedel_code}), each Turing machine may be characterized by some element of the set 
	\begin{align}		\label{eq:DefSetNs}	
									\Ns := \bigcup_{n=1}^{\infty} \big\{ g : \N^{n} \part \N \big\}, 
	\end{align} 
	where we use the symbol "\(\part\)" to denote a \emph{partial mapping}.

	Recursive functions, more specifically referred to as \emph{\(\mu\)-recursive functions}, 
	characterize the notion of computability by means of different approach and were, amongst others, 
	considered by \citeauthor{Kl36} in \cite{Kl36}.
	According to \eqref{eq:DefSetNs}, the set \(\Ns\) contains all possible functions \(g : \N^n \part \N\) for all \(n\in\Nplus\),
	and is thus uncountably infinite in cardinality. The set of those functions \(g \in \Ns\) that correspond to our intuitive understanding of computability in the sense that they can be 
	fully described by a finite sequence of fundamental arithmetic-logic operations, must necessarily be of countable cardinality, and thus be a proper subset of the set \(\Ns\). 
	Contrary to Turing machines, the notion of recursive functions tries to characterize this subset directly by defining a set 
	of fundamental computable operations on the natural numbers, rather than starting from formal languages. Yet, Turing machines and recursive functions turned out to be equivalent in the following 
	sense: the class of functions characterized by the concept of Turing machines coincides with the set of recursive functions \cite{Tu37}. Hence, a function 
	\(g \in \Ns\) can be computed on some Turing machine if and only if it is a recursive function.

	In the following, we will look further into the properties of recursive functions. For \(A\subseteq \N\), denote 
	by \(\Cs(A)\subsetneq \Ns\) the set which consists of the \emph{indicator function} \(\ind_A\) \emph{of} \(A\), the \emph{
	successor function}, and all \emph{constant} and \emph{identity functions} on tuples of natural numbers \cite[Definition~2.1, p.~8]{soare87_book}. 
	By \(\Ca(A)\), denote the closure of \(\Cs(A)\) with respect to \emph{composition}, 
	\emph{primitive recursion} and \emph{unbounded search} \cite[Definition~2.1, p.~8, Definition~2.2, p.~10]{soare87_book}. Then, the set \(\Ca(A)\) is referred
	to as the set of \(A\)-computable functions. In particular, the set \(\Ca(\emptyset)\) is the set of recursive functions. %
	For brevity, we write \(\Ca\) instead of \(\Ca(\emptyset)\). Furthermore, for \(n\in\Nplus\), we denote by \(\Cas{n}(A)\) and \(\Cas{n}\) the set of 
	A-computable functions in \(n\)-variables and the set of recursive functions in \(n\)-variables, respectively. That is, we have
	\begin{align*}
		\Cas{n}(A) :&= \Ca(A) \cap \big\{g : \N^n \part \N\big\}, \\
		\Ca(A) 		&= \bigcup_{n=1}^{\infty} \Cas{n}(A)
	\end{align*}
	for all \(n\in\Nplus\) and all \(A \subseteq \N\).

	\begin{definition} 
		A set \(A\subseteq \N\) is said to be \emph{recursively enumerable} if there exists a recursive function \(g : \N \part \N\) with \emph{domain} \(D(g)\) 
		equal to \(A\).
	\end{definition}
	
	In the context of Turing machines, the domain \(D(g) \subseteq \N\) of a function \(g\in \Cas{1}\) has a dedicated interpretation. 
	Consider a Turing machine \(\TM_g\) that computes the function \(g\). Then, given an input \(m\in\N\), the Turing machine \(\TM_g\) reaches 
	its halting state after a finite number of computational steps if and only if \(m\in D(g)\) is satisfied. 
	In contrast, if \(m\in \N\setminus D(g)\), the Turing machine \(TM_g\) runs forever.
	 
	\begin{definition}
		A set \(G\subseteq \N\) is said to be \emph{recursive} if the corresponding indicator function \(\ind_{G} : \N \tot \{0,1\}\) is a recursive function.
	\end{definition}

	\begin{remark}
		A set \(G \subseteq \N\) is recursive if and only if both \(G\) and \(G^{c} := \N\setminus G\) are recursively enumerable sets. Furthermore, for a set 
		\(G \subseteq \N\), we have \(\Ca(G) = \Ca\) if and only if \(G\) is a recursive set.
	\end{remark}

	The set of \(A\)-computable functions in one variable, \(\Ca_1(A)\),  
	which will be of special significance in the following, is \emph{recursively enumerable} itself. 
	In this context, recursive enumerability refers to the existence of a \emph{universal} \(A\)-computable function \(\Phi^A \in \Ca_2(A)\) such that for all 
	\(A\)-computable functions \(g \in \Ca_1(A)\), there exists a number \(n\in\N\) such that
	\begin{align}	\forall m\in D(g)			&:~\Phi^A(n,m) = g(m), \nonumber\\
								\forall m\notin D(g)	&:~(n,m) \notin D(\Phi^A) 
	\end{align}
	hold true \cite[Theorem~1.5.3, p.~11]{Soare2016}. 	In short, we say that \(g\) satisfies \(\Phi^A(n,m) = g(m)\) for all \(m\in\N\), 
	implicitly (and with some abuse of notation) including the case of \(\Phi^A(n,m)\) and \(g(m)\) being undefined for some \(m\in\N\).
	For all \(n,m \in \N\), define \(\en_n^A(m) := \Phi^A(n,m)\). Then, the family \((\en_n^A)_{n\in\N}\) is a \emph{recursive enumeration} of the set of all \(A\)-computable 
	functions in one variable. %
	A Turing machine \(\TM_{\Phi^\emptyset}\) that computes the function \(\Phi^\emptyset\) is referred to as a 
	\emph{universal Turing machine}.
	
	The universal function \(\Phi^A\) is not unique, and hence, neither is the recursive enumeration \((\en_n^A)_{n\in\N}\). Thus, we
	consider an arbitrary but fixed recursive enumeration \((\en_n^A)_{n\in\N}\) for the rest of this work. For the sake of simplicity, we write \((\en_n)_{n\in\N}\)
	instead of \((\en_n^\emptyset)_{n\in\N}\) in the special case of \(A = \emptyset\). Within the scope of this work, we will consider this case most of the time. 

	The case of \(A \subsetneq \N\) being some non-recursive set leads to the idea of oracle-computations, which we will investigate in Section 
	\ref{sec:OracleComputations}. From the equivalence of Turing machines and recursive functions, we deduce the existence of a recursive 
	\emph{runtime function}, which will be essential in the context of oracle computations. 
	Intuitively speaking, given a universal Turing machine 
	\(\TM_{\Phi^\emptyset}\), we count the number of steps of calculation (that is, the number of successive applications of the specified transformation rules) 
	that are required for \(\TM_{\Phi^\emptyset}\), given an input \((n,m)\in\N^2\), to reach 
	its halting state. If \(m \in D(\en_n)^c\), the counting continues for an infinite amount of time. 	Expressed in a formal way, there exists a (total) recursive
	function \(\Psi : \N^3 \tot \{0,1\}\) such that the following holds true:
	\begin{itemize}
		\item	For all \(n,m,k\in\N\) that satisfy \(\Psi(n,m,k) = 1\) we have 
				\(\Psi(n,m,k+1) = 1.\)
		\item	For all \(n,m\in\N\) that satisfy \(m \in D(\en_n)\), there exists \(k\in\N\) such that 
				we have \(\Psi(n,m,k) = 1\).
		\item 	For all \(n,m\in\N\) that satisfy \(m \notin D(\en_n)\), we have 
				\(\Psi(n,m,k) = 0\) for all \(k\in\N\).
	\end{itemize}
	Even for a fixed enumeration \((\en_n)_{n\in\N}\) of \(\Ca_1\), the runtime function \(\Psi\) is not unique. Hence, we again consider an arbitrary
	but fixed runtime function \(\Psi\) for the remainder of this work.
	
	\begin{definition}\label{def:totality_set}
		For a universal recursive function \(\Phi\) with runtime function \(\Psi\), we denote by 
		\begin{align*}	\At := \{n\in \N : \forall m\in \N : \exists k \in \N : \Psi(n,m,k) = 1\}
		\end{align*} 
		the \emph{totality} set of \(\Phi\). 
	\end{definition}
	
	From the properties of \(\Psi\), it follows that the totality set contains exactly the indices of all total functions in \((\en_n)_{n\in\N}\). 
	Hence, we have 
	\begin{align*}	\At := \{n\in \N : D(\en_n) = \N\},
	\end{align*}
	which yields a direct relation to the family \((D(\en_n))_{n\in\N}\), which enumerates the set of recursively enumerable sets according to
	the universal function \(\Phi\). In this context, we will reconsider the runtime function \(\Psi\) in Lemma \ref{lem:D_in_Sg1} and Lemma 
	\ref{lem:totality_Set_in_Pi2} at the end of Section \ref{sec:hierarchy_real_numbers}. 	Ultimately, the set \(\Ca(\At)\) of \(\At\)-computable
	functions and the set \(\Ca(D(g))\) of \(D(g)\)-computable functions for \(g\in\Ca\) will play a fundamental role in Section \ref{sec:OracleComputations}
	in the scope of oracle computations.

\section{Encoding Abstract Structures into the Natural Numbers}
\label{Sec:GoedelCodes}
	Throughout this work, we consider computations on different abstract structures, like, for example, real numbers 
	and sequences thereof. By the term \enquote{abstract structure}, we refer to sets whose elements are no natural numbers. 
	In this section, we will prepare the formalization of computability on abstract structures, following
	the presentations given in \cite{weihrauch00_book, boolos02_book, soare87_book, Soare2016}.
	
	In the previous section, we have stated that Turing machines are characterized by the set \(\Ca\), which is a subset of 
	the set \(\Ns\). In other words, Turing machines characterize algorithms that operate on the natural numbers. 
	Hence, the elements of an abstract structure are not directly accessible to Turing machines and thus need to
	be represented in a suitable manner. In particular, we want to represent each element of the abstract structure in 
	question by at least one natural number. In formal terms, for an abstract structure \(\As\), we consider partial surjective
	mappings \(\nu_\As : \N \part \As\). We refer to a mapping of this kind as \emph{notation}. 
	
	\begin{remark}\label{rem:goedel_code}
		In \cite[Definition~2.3.1, p.~33]{weihrauch00_book}, \citeauthor{weihrauch00_book} employs the word \enquote{notation} as a name for the concept 
		of describing abstract objects by words of a formal language. For example, every \emph{definable} object (to be precise: every definable set) in ZFC set theory can, 
		by definition, be described by a formula in first-order predicate logic, see \cite[Definition~2.8, p.~26]{Manin10} for details. A formal language may be encoded into the 
		natural numbers, as was done by \citeauthor{Tu37} in \cite{Tu37} in order to prove
		the equivalence of Turing machines and recursive functions, or earlier by \citeauthor{Go31} in \cite{Go31} in the context of his work on incompleteness theorems.
		On the other hand, natural numbers may be denoted by words of a formal language, as is the case for the usual representation of natural numbers by means
		of the Arabic numerals. Ultimately, both the use of formal languages and the use of natural numbers as ``fundamental'' structure lead to the same notion of 
		computability on abstract structures.
	\end{remark}
	
	As indicated at the end 
	of the previous section, we will mostly concern ourselves with recursive functions regarding questions of computability, since they 
	characterize the capabilities of real-world computers. Thus, we will restrict ourselves to considering the set \(\Ca\) 
	within this section. In principle, all of the following considerations apply to general \(A\)-computable functions in the same manner.	

	The set \(\Ca_1\) of recursive functions in one variable yields a direct way to characterize recursive sets. 
	For every recursive indicator function \(\ind_G\), there exists an \(n\in\N\) such that \(\ind_G = \en_n\) holds true.
	Accordingly, we define the notation
	\begin{align}	\label{eq:NotationRecursiveSets}
					n \mapsto G \quad :\Leftrightarrow \quad \en_n = \ind_G~\text{for}~G\subseteq \N. 
	\end{align}
	This notation is (truly) partial in the following sense: a function \(\en_n\) satisfies \(\en_n = \ind_G\) for some
	recursive set \(G\subseteq \N\) if and only if it is total and attains no values other than \(0\) and \(1\). Clearly, there
	exists \(n\in\N\) such that \(\en_n\) does not satisfy these requirements. Hence, only a proper subset of the natural numbers actually represents recursive sets
	with respect to the notation defined in \eqref{eq:NotationRecursiveSets}.

	In the following, we consider \(n\)-tuples \(\tm = (m_j)_{j=1}^{n} = (m_1,m_2,\ldots,m_n) \) of natural numbers and, for \(l\leq k\leq n\), 
	the projection \([\tm]_l^k = (m_j)_{j=l}^{k}\) on the subtuple \((m_j)_{j=l}^{k} = (m_l,\ldots,m_k)\) consisting of those components of 
	\(\tm\) with index between \(l\) and \(k\). In particular, for \(l=k\), we write \([\tm]_k = m_k\). Given tuples \(\tm \in \N^l\) and \(\ts \in \N^k\) 
	for \(l,k\in\N\), we define \(\tm \circ \ts := (m_1,\ldots,m_l,s_1,\ldots,s_k)\in\N^{l+k}\).

	\begin{remark}
		For \(k,n\in\Nplus\) with \(k\leq n\), the function \([\cdot]_k : \N^n \tot \N\) is an element of the set of \emph{identity functions}, 
		which we have previously used to define the set \(\Ca\) of \(\emptyset\)-computable functions. Hence, the 
		function \([\cdot]_k : \N^n \tot \N\) is recursive by definition for all \(k,n\in\Nplus\) with \(k\leq n\).
	\end{remark} 
	
	In order to extend the idea of notations to structures that involve tuples of natural numbers, we make use of the \emph{Cantor pairing function}
	\(\langle\cdot\rangle_2 :\N^2 \tot \N\),
	\begin{align*}	(m_1,m_2) \mapsto m_2 + \frac{1}{2}(m_1+m_2)(m_1+m_2+1),
	\end{align*} 
	which maps the set \(\N^2\) bijectively to the set \(\N\).
	For \(n > 2\), the \(n\)-th extension \(\langle\cdot\rangle_n : \N^n \tot \N\) of the Cantor pairing 
	function is defined inductively by
	\begin{align*}	\langle \tm \rangle_n 	&= \langle m_1,m_2,\ldots m_n\rangle_n \\
										:	&= \langle m_1, \langle [\tm]_{2}^n \rangle_{n-1}\rangle_2\\
											&= \langle m_1, \langle m_2,m_3,\ldots m_n \rangle_{n-1}\rangle_2.
	\end{align*}
	For the sake of completeness, we also define the trivial Cantor \enquote{pairing} 
	\(	\langle\cdot\rangle_1 :\N \tot \N,~ m \mapsto m.
	\)
	For all \(n\in\Nplus\), the function \(\langle\cdot\rangle_n\) is total and for all \(m\in\N\), there exists exactly one \(\tm\in\N^n\) such that
	\(\langle\tm\rangle_n = m\) holds true. Hence, the inverse Cantor pairing function \(\amalg_n : \N \tot \N^n\) is well-defined. 
	Furthermore, \(\langle\cdot\rangle_n\) is recursive, as is \([\amalg_n(\cdot)]_k\) for all \(n,k\in\Nplus\) that satisfy \(k\geq n\). 
	Using the inverse Cantor pairing function
	for \(n = 3\), we can specify a notation for the set of rational numbers \(\Q\) by defining
	\begin{align*}	m \mapsto q := (-1)^{[\amalg_3(m)]_1}\frac{[\amalg_3(m)]_2}{1 + [\amalg_3(m)]_3}
	\end{align*}
	for \(m\in\N\). 

	We define \(\varpi_1(\cdot) := [\amalg_2(\cdot)]_1\) and \(\varpi_2(\cdot) := [\amalg_2(\cdot)]_2\) for the special case of \(n = 2\).
	Then, inverse Cantor pairing function also yields a notation for the set of finite tuples of natural numbers \(\bigcup_{n=1}^{\infty} \N^n\) by setting
	\begin{align*} m \mapsto \amalg_{\varpi_1(m)}\big(\varpi_2(m)\big),
	\end{align*} 
	which is particularly useful whenever we want to define a notation for the set of finite tuples of elements of some abstract structure \(\As\). If \(\As\) 
	admits a notation itself, a notation for \(\bigcup_{n=1}^{\infty} \As^n\) can then be defined by means of composition. 

	For the sake of readability, we write \(\amalg_k \en_n(m)\) instead of \(\amalg_k(\en_n(m))\) in the following. 
	\begin{definition}\label{de:computable-sequence-rationals} 
		An \(n\)-fold sequence of rational numbers \((r_{\tm})_{\tm\in\N^n}\) is said to be \emph{computable} if there exist 
		exists a number \(k\in\N\) such that 
		\begin{align*}	r_\tm = 	(-1)^{[\amalg_3 \en_k(\langle \tm\rangle_n)]_1}
																\frac{[\amalg_3 \en_k(\langle \tm\rangle_n)]_2}{1 + [\amalg_3 \en_k(\langle \tm\rangle_n)]_3}  
		\end{align*}
		holds true for all \(\tm\in\N^n\). 
	\end{definition}

	In general, real numbers defy exact computability by Turing machines due to their irrational and hence infinite nature. 
	Practically relevant functions, like \(\exp\), \(\sin\) and \(\cos\) are not computable exactly, even when their domain is 
	restricted to the rational numbers. Hence, a shift from the domain of exact computability,
	which we have considered so far, to the domain of approximate computability is necessary. In order for approximate computations to be meaningful,
	it is necessary to incorporate a procedure for estimating the approximation error.

	\begin{definition}\label{defeff}
		A sequence \((x_m)_{m\in\N}\) of real numbers is said to converge
		\emph{effectively} towards a number \(x_*\in\R\) if there exists a number \(c\in\N\) such that \(\xi := \en_c\) is a total recursive function 
		and 
		\begin{align}	\label{eq:ModOfConv}
						|x_*-x_m|<\sfrac{1}{2^M}
		\end{align} 
		holds true for all \(m,M\in\N\) that satisfy \(m\geq \xi(M)\).
	\end{definition}

	The function \(\xi : \N \tot \N, M \mapsto \xi(M)\) is referred to as (recursive) \emph{modulus of convergence} for the sequence \((x_m)_{m\in\N}\).

	\begin{definition}\label{de:computable-real}
		A real number \(x\) is said to be \emph{computable} if there exists a computable sequence of 
		rational numbers that converges effectively towards \(x\). 
	\end{definition}

	We denote the set of computable real numbers by \(\Rc\), by \(\Rc^{+0}\) the non-negative
	numbers in \(\Rc\), and by \(\Cc = \Rc + \iu \Rc\) the set of computable complex numbers. Prominent examples of computable, irrational 
	numbers are \(\sqrt{2}\), \(\e\), and \(\pi\). Given a computable real number \(x\), a pair \(\big((r_m)_{m\in\N},\xi\big)\) consisting 
	of a computable sequence \((r_m)_{m\in\N}\) of rational numbers that satisfies \(\lim_{m\to\infty} r_m = x\) and a corresponding 
	recursive modulus of convergence \(\xi\) such that \eqref{eq:ModOfConv} holds true is called a \emph{standard description} of the number \(x\). 
	
	\begin{remark}
		On the set of standard descriptions of computable real numbers, the computable real numbers induce an equivalence relation as follows:
		Two standard descriptions \(\big((r_m)_{m\in\N},\xi\big)\) and \(\big((\rp_m)_{m\in\N},\xi'\big)\) are equivalent,
		denoted by 
		\begin{align*}	\big((r_m)_{m\in\N},\xi\big) \sim \big((\rp_m)_{m\in\N},\xi'\big),
		\end{align*}
		if they represent the same number \(x\in\Rc\). By \(\big((r_m)_{m\in\N},\xi\big)\describes x\), we indicate that \(\big((r_m)_{m\in\N},\xi\big)\)
		is a representative of \(x\in\Rc\). In general, we will employ the symbol '\(\describes\)' in the context of computation on abstract sets
		whenever we want to indicate that some abstract object is represented by some \enquote{less abstract} object which is accessible to Turing machines.
	\end{remark}

	By Definition \ref{de:computable-sequence-rationals}, the family \((\en_k)_{k\in\N}\) induces for all \(n\in\Nplus\) a notation for the set 
	of \(n\)-fold computable sequences of rational numbers. Since Definition \ref{de:computable-sequence-rationals} implicitly requires the 
	function \(\en_k\) to satisfy \(D(\en_k) = \N\), the notation is partial. Furthermore, the mapping 
	\begin{align} 	\label{eq:DefNotationRc}
					l \quad \mapsto \quad (\varpi_1(l), \varpi_2(l)) =: (k,c) \quad \mapsto \quad (\en_k,\en_c), 
	\end{align}
	provides a notation for the set \(\Rc\). For every standard description \(\big((r_m)_{m\in\N},\xi\big)\) of some computable number \(x\), 
	there exist \(k,c \in \N\), such that \(\en_k\) characterizes the sequence \((r_m)_{m\in\N}\) in the sense of Definition 
	\ref{de:computable-sequence-rationals} and \(\en_c\) characterizes the function \(\xi\) in the sense of Definition \ref{defeff}.

	\begin{definition}\label{def:CSCN}
		An \(n\)-fold sequence \((x_\tm)_{\tm\in\N^n}\) of computable real numbers is called \emph{computable} if 
		there exists an \((n+1)\)-fold computable sequence \((r_{\tm\circ s})_{\tm\circ s\in\N^{n+1}}\) of rational 
		numbers as well as a number \(c\in\N\) such that \(\xi := \en_c\) is a total recursive function and 
		\begin{align*}	\big|x_{\tm} - r_{\tm\circ s}\big| < \frac{1}{2^M}
		\end{align*}
		holds true for all \(s\in\N\), \(\tm\in\N^{n}\), \(M\in\N\) that satisfy  
		\(s \geq \xi(\langle \tm \circ M\rangle_{n+1})\).
	\end{definition}

	The pair \(\big((r_{\tm\circ s})_{\tm\circ s\in\N^{n+1}}, \xi\big)\) is referred to as a \emph{standard description} of the sequence \((x_\tm)_{\tm\in\N^n}\).
		Again, the mapping \(l \mapsto (\en_k,\en_c)\) defined in \eqref{eq:DefNotationRc} provides a notation for set of \(n\)-fold 
	computable sequences of computable numbers. Every standard description \(\big((r_{\tm\circ s})_{\tm\circ s\in\N^{n+1}}, \xi\big)\) 
	of some sequence \((x_\tm)_{\tm\in\N^n}\) may be characterized by a pair \((\en_k,\en_c)\) according to Definitions
	\ref{de:computable-sequence-rationals} and \ref{defeff}, analogous to the notation for the set \(\Rc\) defined above.

	Finally, returning to the domain of exact computation, we introduce for all \(n\in \Nplus\) a notation for the set 
	\(\Cas{n}\) of recursive functions in \(n\) variables by defining
	\begin{align}\label{eq:enumC}		m \mapsto \en_{m}\big(\langle \cdot \rangle_{n}\big).
	\end{align}
	Clearly, \(\en_{m}\big(\langle \cdot \rangle_{n}\big) : \N^{n} \part \N\) is a recursive function for all \(m\in\N\).
	On the other hand, consider an arbitrary recursive function \(g :\N^n \part \N\). Then, \(g(\amalg_n(\cdot))\) is a recursive 
	function in one variable. Consequently, there exists an \(m\in\N\) such that \(\en_{m}(\cdot) = g(\amalg_n(\cdot))\) 
	holds true and we have
	\begin{align*} 		\en_{m}\big(\langle \cdot \rangle_{n}\big) = g\big(\amalg_n(\langle \cdot \rangle_{n})\big) = g\big(\cdot\big).
	\end{align*}
	Hence, \eqref{eq:enumC} yields for all \(n\in\Nplus\) a notation for the set \(\Cas{n}\). By setting
	\begin{align*}		m \mapsto \en_{\varpi_1(m)}\big(\langle \cdot \rangle_{\varpi_2(m)}\big),
	\end{align*}
	we extend this notation to a notation for the set \(\Ca\) of all recursive functions.

\section{Computations and Algorithms on Abstract Structures}
	In the previous section, we have introduced the notion of notations. In this section, we will employ notations to 
	formally define the idea of \emph{computability on a Turing machine} for abstract structures.
	 
	Let \(\As\) and \(\As'\) be two abstract structures with fixed notations \(\nu_\As : \N \part \As\) and 
	\(\nu_{\As'} : \N \part \As'\). That is, \(\nu_\As\) and \(\nu_{\As'}\) are (not necessarily total) 
	surjections from \(\N\) onto \(\As\) and \(\As'\), respectively. Furthermore, consider a mapping \(\Gop : \As \part \As'\).
	If there exists a recursive function \(g : \N \part \N\) such that for all pairs \((a,n)\in \As \times \N\), we have 
	\begin{align*}	a\in D(\Gop)\wedge\nu_\As(n) = a \quad\Rightarrow\quad \nu_{\As'}(g(n)) = \Gop(a),
	\end{align*}
	then we say that \emph{there exists a Turing machine \(\TM_\Gop\) that returns \(a' := \Gop(a)\) for input \(a\).}
	This Turing machine computes the mapping \(a \mapsto a' := \Gop(a)\) in the following sense: 
	\begin{itemize}
		\item Whenever \(\TM_\Gop\) is presented with a number \(n\) that \emph{denotes} the object \(a = \nu_{\As}(n) \in D(\Gop)\) with
				respect to the notation \(\nu_{\As}\), it returns a number \(g(n)\) that \emph{denotes} the object \(a' = \nu_{\As'}(g(n)) \in \As'\)
				with respect to the notation \(\nu_{\As'}\). 
		\item If \(\Gop\) is undefined at the point \(a = \nu_{\As}(n) \in \As\), then either \(\TM_\Gop\) does not halt in a finite number of steps 
					for input \(n\), or \(g(n)\) is \emph{not} an element of \(D(\nu_{\As'})\). 
		\item In particular, \(\TM_\Gop\) maps the set of numbers that denote the object \(a\in\As\) with respect to \(\nu_{\As}\)
				to a subset of the set of numbers that denote the object \(\Gop(a)\in\As'\) with respect to \(\nu_{\As'}\),
				i.e., for all \(a\in D(\Gop)\), we have
				\begin{align*}	g(\{n : \nu_\As(n) = a\}) \subseteq \{n : \nu_{\As'}(n) = \Gop(a)\}.
				\end{align*}
	\end{itemize}
	With some abuse of notation, we also write \(\TM_\Gop : \As \rightarrow \As', a \mapsto \TM_\Gop(a) = a'\), 
	despite the fact that \(\TM_\Gop\) actually computes a mapping
	on natural numbers rather than a mapping on the abstract structures \(\As\) and \(\As'\).

	\begin{remark} If both \(\As\) and \(\As'\) are either the set of computable reals \(\Rc\) or the set of computable 
		complex numbers \(\Cc\), then a function \(f : \As \part \As'\) that satisfies the above notion of computability 
		is referred to as \emph{Markov computable}. Markov computability is not the strongest
		notion of computability on \(\Rc\) and \(\Cc\) that has been used in the literature: \emph{Turing computability}, 
		which is applicable to functions on all of 
		\(\R\) and \(\C\), implies Markov computability when restricted to \(\Rc\) and \(\Cc\). For a comparison, see \cite[Appendix~2.9, p.~21]{avigad14}. 
	\end{remark}

	The majority of notations introduced in Section \ref{Sec:GoedelCodes} is induced by the enumeration \((\en_n)_{n\in\N}\). 
	If the notations of two abstract structures \(\As\) and \(\As'\) are induced by a recursive enumeration of \(\Ca_1\),
	then any mapping \(\Gop : \As \part \As'\) is, in principle, a transformation of recursive functions. 
	In the following, we will introduce a consequence of the \emph{s-m-n Theorem} \cite[Theorem~3.5, p.~16]{soare87_book}, 
	which captures a fundamental structural property of the set \(\Ca\). In particular, we will characterize the conditions
	under which a transformation of recursive functions is computable on a Turing machine.
	
	Consider the following case: for a tuple \(\tm \in \N^n\), the abstract structure \(\As\) equals the set \(\Cas{m_1} \times \cdots \times \Cas{m_n}\),
	while the abstract structure \(\As'\) equals the set \(\Ca\). For all \(j \in \{1,\ldots,n\}\), the set \(\Cas{m_j}\) is equipped with an individual notation. 
	We can extend the individual notations to a joint notation \(\nu_{\As}\) for the set \(\Cas{m_1} \times \cdots \times \Cas{m_n}\)
	by first enumerating all \(n\)-tuples of natural numbers and then applying the individual notations to the respective component. 
	Furthermore, for the set \(\Ca\), a notation \(\nu_{\As'}\) was introduced at the end of Section \ref{Sec:GoedelCodes}. Now, consider a mapping
	\begin{align*}\begin{array}{rclcl}
		\Gop 	& : & \Cas{m_1} \times \cdots \times \Cas{m_n} 	& \rightarrow & \Ca, \\
					&		&	(g_1,g_2,\ldots,g_n)											& \mapsto 		& h,
	\end{array}\end{align*}
	which, for some number \(k\in\Nplus\), is of the form
	\begin{align*}\begin{array}{llcl}
		(g_j)_{j=1}^{n} 												& = : \th_0			& \mapsto 		&h_1, \\
		(g_j)_{j=1}^{n}\circ(h_j)_{j=1}^{1}			& = : \th_1			& \mapsto 		&h_2, \\
																						& 							& \vdots  		&\\
		(g_j)_{j=1}^{n}\circ(h_j)_{j=1}^{k-1}		& = : \th_{k-1}	& \mapsto			&h_k =: h, \\																																							
	\end{array}\end{align*}
	such that for all \(j\in\{1,\ldots,k\}\), the function \(h_j\) is either an
	element of \(\Cs(\emptyset)\) or emerges from \(\th_{j-1}\) through concatenation, primitive recursion
	or unbounded search. Then, there exists a mapping \(g' \in \Ca_1\) such that
	for all \(l\in\N\) and all \((g_j)_{j=1}^{n} \in \Cas{m_1} \times \cdots \times \Cas{m_n}\) that satisfy
	\(\nu_\As(l) = (g_j)_{j=1}^{n}\), we have \(\nu_{\As'}(g'(l)) = h\). That is, there exists a Turing machine
	\(\TM_\Gop\) that computes the mapping \((g_1,g_2,\ldots,g_n) \mapsto h\). 
	
	As a rule of thumb, if we can implement a mapping \(\Gop : \Cas{m_1} \times \cdots \times \Cas{m_n} \rightarrow  \Ca\) on a real world computer,
	it is computable by a Turing machine, which applies to all common arithmetic and logic operations. Throughout the rest of this work,
	we make implicit use of this principle on several occasions.

\section{The Arithmetical Hierarchy of Real Numbers}
\label{sec:hierarchy_real_numbers}
	Since \citeauthor{turing36} published his work on the theory of computation, it has been known that almost all real numbers are uncomputable. In an attempt to 
	characterize different degrees of (un)computability, \emph{\citeauthor{zw01}} 
	introduced the arithmetical hierarchy of real numbers \cite{zw01}, which is strongly 
	related to the \emph{Kleene–Mostowski hierarchy} of subsets of natural numbers \cite{Kl43, Mo47}. 
	The hierarchical level of a real number solely depends on the logical structure that is used to define the number. 

	\begin{definition}[Kleene–Mostowski Hierarchy, cf. \cite{Kl43, Mo47}]\label{def:KlMo}
		For \(n\in\Nplus\), consider \(\tm = (m_1,\ldots,m_n) \in\N^n\). Then, the sets \(\Sigma^0_n \subsetneq 2^{\N}\), \(\Pi^0_n \subsetneq 2^{\N}\) and
		\(\Delta^0_n \subsetneq 2^{\N}\) are defined as follows:
		\begin{itemize}
			\item A set \(A\subseteq \N\) satisfies \(A\in\Sigma_n^0\) if there exists a recursive set \(G\subseteq \N\) such that for all \(j\in \N\), we have
				\(j\in A\) if and only if
				\begin{align*}
					\exists m_1 \forall m_2 \exists m_3 \ldots \qnt m_n(\langle \tm\circ j\rangle_{n+1} \in G)
				\end{align*}
				holds true, where \enquote{\(\qnt\)} is replaced by \enquote{\(\forall\)} if \(n\) is even and by \enquote{\(\exists\)} if \(n\) is odd.
			\item A set \(A\subseteq \N\) satisfies \(A\in\Pi_n^0\) if there exists a recursive set \(G\subseteq \N\) such that for all \(j\in \N\), we have
				\(j\in A\) if and only if
				\begin{align*}
					\forall m_1 \exists m_2 \forall m_3 \ldots \qnt m_n(\langle \tm\circ j\rangle_{n+1} \in G)
				\end{align*}
				holds true, where \enquote{\(\qnt\)} is replaced by \enquote{\(\exists\)} if \(n\) is even and by \enquote{\(\forall\)} if \(n\) is odd.
			\item A set \(A\subseteq \N\) satisfies \(A\in\Delta_n^0\) if satisfies both \(A\in\Sigma_n^0\) and \(A\in\Pi_n^0\). Hence, we have
				\(\Delta_n^0 = \Sigma_n^0 \cap \Pi_n^0\).
		\end{itemize}
	\end{definition}

	\begin{definition}[Zheng-Weihrauch Hierarchy, cf. \cite{zw01}] 
		For \(n\in\Nplus\), consider \(\tm = (m_1,\ldots,m_n) \in\N^n\). Then, the sets \(\Sigma_n \subsetneq \R\), \(\Pi_n \subsetneq \R\) and
		\(\Delta_n \subsetneq \R\) are defined as follows:
		\begin{itemize}
			\item A number \(x_*\in\R\) satisfies \(x_*\in\Sigma_n\) if there exists an \(n\)-fold computable sequence \((r_{\tm})_{\tm\in\N^n}\)
				of rational numbers such that 
				\begin{align*}
					x_* = \sup_{m_1\in\N}\inf_{m_2\in\N}\sup_{m_3\in\N}\ldots\mima_{m_n\in\N}~ (r_{\tm})
				\end{align*}
				holds true, where \enquote{\(\mima\)} is replaced by \enquote{\(\inf\)} if \(n\) is even and by \enquote{\(\sup\)} if \(n\) is odd.
			\item A number \(x_*\in\R\) satisfies \(x_*\in\Pi_n\) if there exists an \(n\)-fold computable sequence \((r_{\tm})_{\tm\in\N^n}\)
				of rational numbers such that 
				\begin{align*}
					x_* = \inf_{m_1\in\N}\sup_{m_2\in\N}\inf_{m_3\in\N}\ldots\mima_{m_n\in\N}~ (r_{\tm})
				\end{align*}
				holds true, where \enquote{\(\mima\)} is replaced by \enquote{\(\sup\)} if \(n\) is even and by \enquote{\(\inf\)} if \(n\) is odd.
			\item A number \(x_*\in\R\) satisfies \(x_*\in\Delta_n\) if satisfies both \(x_*\in\Sigma_n\) and \(x_*\in\Pi_n\). Hence, we have
				\(\Delta_n = \Sigma_n \cap \Pi_n\). 
		\end{itemize}
	\end{definition}

	\begin{remark}
		A real number \(x\) is computable if and only if it satisfies both \(x \in \Pi_1\) and \(x \in \Sigma_1\). Hence, we have \(\Rc = \Delta_1\). 
	\end{remark}

	Given a set \(A\subseteq \N\), we denote \(x[A] := \sum_{j\in A}\sfrac{1}{2^{(j+1)}}\). \citeauthor{zw01} 
	showed that for all \(n\in\Nplus\), if \(A\) satisfies \(A \in \Sigma_n^0\), \(A \in \Pi_n^0\) or \(A \in \Delta_n^0\), then
	\(x[A]\) satisfies \(x[A] \in \Sigma_n\), \(x[A] \in \Pi_n\) or \(x[A] \in \Delta_n\), respectively.

	Assuming the finiteness of the respective suprema and infima, an \(n\)-fold computable sequence \((x_{\tm})_{\tm\in\N^n}\) 
	of computable numbers is referred to as \(n\)-th order
	lower Zheng-Weihrauch (ZW) description of the real number 
	\begin{align*}
		x_* = \sup_{m_1\in\N}\inf_{m_2\in\N}\sup_{m_3\in\N}\ldots\mima_{m_n\in\N}~ (x_{\tm}),
	\end{align*}
	where \enquote{\(\mima\)} is replaced by \enquote{\(\inf\)} if \(n\) is even and by \enquote{\(\sup\)} if \(n\) is odd.

	Likewise, again assuming the finiteness of the respective suprema and infima, an \(n\)-fold computable sequence 
	\((x_{\tm})_{\tm\in\N^{n}}\) of computable numbers is referred to as \(n\)-th order
	upper Zheng-Weihrauch (ZW) description of the real number 
	\begin{align*}
		x_* = \inf_{m_1\in\N}\sup_{m_2\in\N}\inf_{m_3\in\N}\ldots\mima_{m_n\in\N}~ (x_{\tm}),
	\end{align*}
	where \enquote{\(\mima\)} is replaced by \enquote{\(\sup\)} if \(n\) is even and by \enquote{\(\inf\)} if \(n\) is odd.

	\begin{lemma}\label{lem:LowerZWImpliesPi2}
		If there exists an \(n\)-th order upper ZW description for a number \(x_* \in \R\), then \(x_*\in\Pi_n\) is satisfied.
	\end{lemma}\begin{proof}	 
		Consider the standard description \(\big((\rp_{\ts})_{\ts\in\N^{n+1}},\xi\big)\) of an \(n\)-th order upper ZW 
		description \((x_\tm)_{\tm\in\N^n}\) of \(x_*\) and define
		\begin{align*}
			g'(\tm)	:&= \big(\varpi_1(m_1),m_2,\ldots,m_n\big), \vphantom{\frac{1}{2^{\varpi_2(m_1)}}}\\
			g(\tm)  :&= g'(\tm)\circ\xi\big(g'(\tm),\varpi_2(m_1)\big), \vphantom{\frac{1}{2^{\varpi_2(m_1)}}}\\
			r_{\tm} :&= \rp_{g(\tm)} + \frac{1}{2^{\varpi_2(m_1)}},
		\end{align*}
		for all \(\tm\in\N^n\).
		Then, \((r_{\tm})_{\tm\in\N^{n}}\) is an \(n\)-fold computable sequence of rational numbers that satisfies 
		\(r_{\tm} \geq x_{g'(\tm)}\) for all \(\tm\in\N^n\), and hence
		\begin{align}
			x_* &\leq \sup_{m_2\in\N}\inf_{m_3\in\N}\ldots\mima_{m_n\in\N} \big(x_{g'(\tm)}\big) \nonumber\\ 
					&\leq	\sup_{m_2\in\N}\inf_{m_3\in\N}\ldots\mima_{m_n\in\N} \big(r_{\tm}\big)
					\label{eq:lem:LowerZWImpliesPi2::I}
		\end{align}
		for all \(m_1\in\N\) as well. Since \(\big((\rp_{\ts})_{\ts\in\N^{n+1}},\xi\big)\) 
		is a standard description of \((x_\tm)_{\tm\in\N}\), we have
		\begin{align*}
			&\bigg|\sup_{m_2\in\N}\inf_{m_3\in\N}\ldots\mima_{m_n\in\N}\left(x_{\tm}\right)\ldots\\
						&\qquad \ldots - \sup_{m_2\in\N}\inf_{m_3\in\N}\ldots\mima_{m_n\in\N}\left(\rp_{\tm\circ\xi((\tm),M)}\right)\bigg| < \frac{1}{2^M}
		\end{align*}
		for all \(m_1,M\in\N\). Since furthermore, \((x_\tm)_{\tm\in\N}\) is an upper ZW description of \(x_*\), we conclude that for \(\epsilon > 0\) arbitrary,
		there exist \(l,k\in\N\) such that for \(\tm' := (l,m_2,\ldots,m_n)\in\N^{n}\), 
		\begin{align*}
			x_* + \epsilon \geq \sup_{m_2\in\N}\inf_{m_3\in\N}\ldots\mima_{m_n\in\N}\left(\rp_{\tm'\circ\xi((\tm'),k)}\right)
			-\frac{1}{2^k}
		\end{align*}
		holds true. The mapping \(\langle\cdot\rangle_2^{-1} : \N \tot \N^2\) is bijective. Hence, for all \(l,k\in\N\), there exists
		\(m_1\in\N\) such that \(\langle m_1\rangle_2^{-1} = (\varpi_1(m_1),\varpi_2(m_1)) = (l,k)\) holds true. Consequently, there exists \(m_1\in\N\)
		such that 
		\begin{align}
			x_* + \epsilon \geq \sup_{m_2\in\N}\inf_{m_3\in\N}\ldots\mima_{m_n\in\N} \big(r_{\tm}\big)
			\label{eq:lem:LowerZWImpliesPi2::II}
		\end{align}
		holds true. Joining \eqref{eq:lem:LowerZWImpliesPi2::I} and \eqref{eq:lem:LowerZWImpliesPi2::II}, we conclude that 
		\begin{align*}
			x_* + \epsilon \geq \inf_{m_1\in\N}\sup_{m_2\in\N}\sup_{m_3\in\N}\ldots\mima_{m_n\in\N} \big(r_{\tm}\big) \geq x_*
		\end{align*} 
		is satisfied. Since \(\epsilon > 0\) was chosen arbitrarily, the claim follows.
	\end{proof}

	\begin{remark}
		From the definition of computable sequences of computable numbers, it is imminent that every computable sequence of rational numbers
		is also a computable sequence of computable numbers. Hence, a converse to Lemma \ref{lem:LowerZWImpliesPi2} is straightforward to prove,
		and we conclude that if a number \(x_*\in\R\) satisfies \(x_* \in \Pi_n\), there exists an \(n\)-th order upper ZW description for \(x_*\). 
		Likewise, if a number \(y_*\in\R\) satisfies \(y_* \in \Sigma_n\), there exists an \(n\)-th order lower ZW description for \(y_*\).
		Furthermore, the line of reasoning presented in the proof of Lemma \ref{lem:LowerZWImpliesPi2} applies analogously 
		to \(n\)-th order lower ZW description of some number \(y_* \in \R\).
		In summary, we arrive at the following: a number \(x_*\in\R\) satisfies \(x_* \in \Pi_n\) if
		and only if there exists an \(n\)-th order upper ZW description for \(x_*\), while a number \(y_*\in\R\) satisfies \(y_* \in \Sigma_n\) if
		and only if there exists an \(n\)-th order lower ZW description for \(y_*\).
	\end{remark}

	The following lemma was stated in a slightly different form by \citeauthor{zw01}. A close examination shows that the corresponding proof is constructive
	and exclusively employs operations that can be computed on a Turing machine. Hence, the derivation provided by \citeauthor{zw01} is sufficient to prove 
	the subsequent version of the lemma.

	\begin{lemma}[{{\cite[Lemma~3.2, p.~55]{zw01}}}]\label{lem:ZhengWeihrauchI}
		Let \((r_{\tm})_{\tm\in\N^2}\) be a computable double sequence of rational numbers such that 
		\(\inf_{m_1\in\N}\sup_{m_2\in\N}(r_{\tm})\) exists. There exists a Turing machine which computes a mapping
		\((r_{\tm})_{\tm\in\N^2} \mapsto (\rp_{m})_{m\in\N}\) such that \((\rp_{m})_{m\in\N}\)
		is a computable sequence of rational numbers and \(\limsup_{m\to\infty} (\rp_{m}) = \inf_{m_1\in\N}\sup_{m_2\in\N}(r_{\tm})\)
		is satisfied. \qed
	\end{lemma}
	
	In Section \ref{sec:prel_on_TM}, we introduced the totality set \(\At\) and its relation to the family \((D(\en_n))_{n\in\N}\).
	As indicated, these sets will play a role in the context of oracle computations. In the following, we will characterize 
	\(\At\) and \((D(\en_n))_{n\in\N}\) with respect to the Kleene–Mostowski hierarchy.
	
	\begin{lemma}	\label{lem:D_in_Sg1}
		For all \(n\in\N\), we have \(D(\en_n) \in \Sigma_1^0\).
	\end{lemma}\begin{proof}
		Define \(G^\Psi_n := \{m \in \N : \Psi (n, \varpi_2(m), \varpi_1(m)) = 1\}\). 
		Since \(\Psi\) is a total recursive function, the set \(G^\Psi_n\) is recursive. Furthermore, incorporating the definition of \(\Psi\), we have
		\begin{align*}	D(\en_n) 	&=\{ j \in \N : \exists k \in \N : \Psi (n,j,k) = 1\} \\
															&= \{ j \in \N : \exists k \in \N : \langle k, j \rangle_2 \in G^\Psi_n \}.
		\end{align*}
		The claim then follows from Definition \ref{def:KlMo}.
	\end{proof}
	
	\begin{lemma}	\label{lem:totality_Set_in_Pi2}
		The totality set \(\At\) satisfies \(\At \in \Pi_2^0\).
	\end{lemma}\begin{proof}
		Again, the claim follows from the existence of a suitable recursive set, which is induced
		by the runtime function \(\Psi\). Define 
		\begin{align*}	G^\Psi := \big\{ m \in \N : \Psi\big([\amalg_3(m)]_3, [\amalg_3(m)]_1, [\amalg_3(m)]_2\big) = 1\big\}.
		\end{align*}
		Since \(\Psi\) is a total recursive function, the set \(G^\Psi\) is recursive. Furthermore,
		by Definition \ref{def:totality_set}, we have
		\begin{align*}	\At &= \{n\in \N : \forall j\in \N : \exists k \in \N : \Psi(n,j,k) = 1\} \\
												&= \{n\in \N : \forall j\in \N : \exists k \in \N : \langle j,k,n \rangle_3 \in G^\Psi\}.
		\end{align*}
		The claim then follows from Definition \ref{def:KlMo}.
	\end{proof}
	
\section{Computable Bandlimited Signals}
\label{sec:ComputableBandlimitedSignals}
	Having introduced a framework for computable analysis, we are now equipped to formalize the
	concept of computable bandlimited signals. Again, for the sake of self-containedness, we repeat some
	of the definitions found in \cite{boche21d_accepted}.
	
	\begin{definition}\label{de:elementary-computable}
		We call a signal \(f\) elementary computable if there exists a natural number \(L\) and a
		sequence of computable numbers \((c_k)_{k=-L}^{L}\) that satisfy
		\begin{align*}
			f(t)
			=
			\sum_{k=-L}^{L} c_k \frac{\sin(\pi(t-k))}{\pi(t-k)} .
		\end{align*}
	\end{definition}
	
	The building blocks of an elementary computable signal are \(\sinc\) functions. 
	Hence, elementary computable signals are exactly those functions that can be represented
	by a finite Shannon sampling series with computable coefficients \((c_k)_{k=-L}^{L}\).
	Note that every elementary computable signal \(f\) is a finite sum of computable
	continuous functions and hence a computable continuous function.
	As a consequence, for every \(t \in \Rc\) the number \(f(t)\) is computable.
	Further, the sum of finitely many elementary computable signals is elementary
	computable, as well as the product of an elementary computable signal with a computable
	number.

	\begin{definition}
		A signal in \(f \in \Bs_{\pi}^p\), \(p \in (1,\infty) \cap \Rc\), is called computable in
		\(\Bs_{\pi}^p\) if there exists a computable sequence of elementary computable signals
		\((f_m)_{m \in \N}\) and a recursive function \(\xi \colon \N \to \N\) such that for all
		\(M \in \N\) we have
		\begin{align*}
			\lVert f - f_{m} \rVert_{\Bs_{\pi}^p}
			\leq
			\frac{1}{2^M}
		\end{align*}
		for all
		\(m \geq \xi(M)\).
		By \(\CBs_{\pi}^p\), \(p \in (1,\infty) \cap \Rc\), we denote the set of all signals in
		\(\Bs_{\pi}^p\) that are computable in \(\Bs_{\pi}^p\).
	\end{definition}
	
	According to this definition we can approximate any signal \(f \in \CBs_{\pi}^p\),
	\(p \in (1,\infty) \cap \Rc\), by an elementary computable signal, where we have an
	``effective'', i.e. 
	computable control of the approximation error.
	For every prescribed approximation error \(\epsilon>0\), \(\epsilon \in \Rc\), we can compute
	a number \(M \in \N\) such that \( M \geq -\log_2(\epsilon)\) holds true. Hence, the approximation error
	\(\lVert f - f_{m} \rVert_{\Bs_{\pi}^p}\) is less than or equal to \(\epsilon\) for all
	\(m \geq \xi(M)\).

	We finally give the definition of computability for bandlimited signals in
	\(\Es_{\pi}\).
	
	\begin{definition}\label{def:CEs_pi}
		We call a signal \(f \in \Es_{\pi}\) a computable bandlimited signal if the coefficients
		\((a_n)_{n \in \N}\) of the Taylor series \eqref{eq:power-series} form a computable
		sequence of computable numbers.
		By \(\CEs_{\pi}\) we denote the set of all signals in \(\Es_{\pi}\) that are computable.
	\end{definition}

	Note that a signal \(f \in \CEs_{\pi}\) is completely determined by the computable
	sequence \((a_n)_{n \in \N}\) of computable numbers according to the representation
	\eqref{eq:power-series}.
	Hence, a program for \((a_n)_{n \in \N}\) gives a complete description of \(f\).
	
	For \(f_1,f_2 \in \CEs_{\pi}\) and \(\alpha_1,\alpha_2 \in \Cc\), we have
	\(\alpha_1 f_1 + \alpha_2 f_2 \in \CEs_{\pi}\), i.e., \(\CEs_{\pi}\) has a linear structure.

	\begin{remark}
		We have \(\CBs_{\pi}^1 \subsetneq \CBs_{\pi}^2 \subsetneq \CEs_{\pi}\),
		which shows that \(\CEs_{\pi}\) is the largest of these three sets.
	\end{remark}

	\begin{remark}
		For \(f \in \Bs_{\pi}^2\), it follows from the properties of the Fourier transform \(\Ft{f}\)
		that the family \((a_n)_{n\in\N}\) of Taylor coefficients satisfies
		\begin{align*}	a_n = \frac{1}{2\pi} \int_{-\pi}^{\pi} \Ft{f}(\omega) (\iu \omega)^n \di{\omega}
		\end{align*}
		for all \(n\in\N\). If \(f\) furthermore satisfies \(f \in \CBs_{\pi}^2\), the Fourier transform \(\Ft{f}\) is a computable signal in
		\(L^2[-\pi,\pi]\) and the family of coefficients \((a_n)_{n \in \N}\) forms a computable sequence of computable numbers.
	\end{remark}
	
	\begin{definition}\label{def:WE_pi}
		Let \(f\) satisfy \(f\in\Es_\pi\). If there exists a computable double sequence \((p_\tm)_{\tm\in\N^2}\) of rational polynomials
		as well as a recursive function \(\xi : \N^2 \rightarrow \N\) such that
		\begin{align*}	\left| f(z) - p_{m_1,m_2}(z) \right| < \frac{1}{2^M}
		\end{align*}
		holds true for all \(z\in\C,m_1,m_2,M \in \N\) that satisfy \(|z| \leq m_2\) and \(m_1 \geq \xi(M,m_2)\), then \(f\) is referred to 
		as \emph{Weierstrass effective}. We denote \(\WEs_\pi\) the set of Weierstrass effective \(\pi\)-bandlimited entire functions.  
	\end{definition}

	\begin{theorem}\label{thm:CEs_equals_WEs}
		The sets \(\CEs_\pi\) and \(\WEs_\pi\) coincide.
	\end{theorem}\begin{proof}
		We start by proving that \(f \in \CEs_\pi \Rightarrow f \in \WEs_\pi\) holds true
		for all \(f \in \CEs_\pi\). Let \(f\in\CEs_\pi\) satisfy \eqref{eq:power-series}, where
		\((a_n)_{n\in\N}\) is a computable sequence of computable numbers with standard description
		\(((r_{n,m})_{n,m\in\N},\xi)\). Without loss of generality, assume that \(\xi\) is the identity function, i.e.,
		\((r_{n,m})_{n,m\in\N}\) and \((a_n)_{n\in\N}\) satisfy
		\begin{align*} |a_n -r_{n,m}| < \frac{1}{2^m}
		\end{align*}
		for all \(n,m\in\N\). Since \(f \in \CEs_\pi\) holds true, there exists \(L\in\N\) such that 
		\begin{align*} \sqrt[n]{|a_n|} \leq L
		\end{align*}
		is satisfied for all \(n\in\N\). For \(K,m\in\N,z\in\C\), define
		\begin{align*}		\overline{f}_{K,m}(z) :		&= f(z) - \sum_{n=0}^{K}\frac{r_{n,m}}{n!}{t^n} \\
																								&= \left(\sum_{n=0}^{K} \frac{a_n-r_{n,m}}{n!}t^n\right) 
																									+ \left(\sum_{n=K+1}^{\infty} \frac{a_n}{n!}t^n\right).
		\end{align*}
		The function \(\overline{f}_{K,m}\) characterizes the error that arises when approximating the signal \(f\) by
		a computable polynomial based on the standard description \(((r_{n,m})_{n,m\in\N},\xi)\) of \((a_n)_{n\in\N}\). 
		We will prove the first part of the theorem by deriving a computable upper bound on \(\overline{f}_{K,m}\).
		Set \(\overline{a}_{n,m} := a_n-r_{n,m}\) for all \(n,m\in\N\). We have
		\begin{align} 	\max_{|z| \leq J}\left|\overline{f}_{K,m}(z)\right| 
											&\leq \left(\sum_{n=0}^{K} \frac{|\overline{a}_{n,m}|}{n!}J^n\right) + \left(\sum_{n=K+1}^{\infty} \frac{L^n}{n!}J^n\right)\nonumber \\
											&\leq \left(\sum_{n=0}^{K} \frac{\sfrac{1}{2^m}}{n!}J^n\right) + \left(\sum_{n=K+1}^{\infty} \frac{L^n}{n!}J^n\right)
												\label{eq:thm:CEs_equals_WEs::fooI} \\
											&\leq \frac{1}{2^m}4^J + \sum_{n=K+1}^{\infty} \left(\frac{LJ}{(\sfrac{1}{4})n}\right)^n \nonumber\\
											&= 		2^{2J - m} + \sum_{n=K+1}^{\infty} \left(\frac{4LJ}{n}\right)^n \nonumber\\
											&\leq	2^{2J - m} + \sum_{n=K+1}^{\infty} \bigg(\underbrace{\frac{4LJ}{K+1}}_{=: x(L,J)}\bigg)^n, \nonumber
		\end{align}
		where the bound on the first sum in \eqref{eq:thm:CEs_equals_WEs::fooI} follows from the series expansion of the function \(x \mapsto \e^x\)
		and the bound on the second sum in \eqref{eq:thm:CEs_equals_WEs::fooI} follows from the inequality \(n! \geq (\sfrac{1}{4}\cdot n)^n\), which 
		in turn follows from \emph{Stirling's approximation} for factorials. The remaining sum converges if \(x(L,J)\) is sufficiently close to \(0\).
		In particular, 
		\begin{align*}	\sum_{n=K+1}^{\infty} x^{n} = x^{K+1}\frac{1}{1-x}	
		\end{align*}
		holds true for all \(x\in[0,1)\). Hence, consider \(K\in\N\) such that \(K+1 > 8LJ\) is satisfied. Then, we have
		\begin{align*} \max_{|z| \leq J}\left|\overline{f}_{K,m}(z)\right| \leq 2^{2J-m} + 2^{-K}.
		\end{align*}
		For \(M,J\in\N, z\in\C\), we define
		\begin{align*}		N(M,J) 			:		&= 2J + M + 2, \\
											K(M,J) 			: 	&= \max\{8LJ, M + 2\}, \\
											p_{M,J}(z)	:		&= \sum_{n = 0}^{K(J,M)} \frac{r_{n,N(J,M)}}{n!}z^n.
		\end{align*}
		Then, \((p_{M,J})_{M,J\in\N}\) is a computable family of polynomials such that
		\begin{align*} \left|f(z) - p_{M,J}(z)\right| < \frac{1}{2^{M}}
		\end{align*}
		holds true for all \(M,J\in\N,z\in\C\) that satisfy \(|z| \leq J\). Hence, \(f\) satisfies 
		\(f\in\WEs_\pi\).
		
		It remains to show that \(f \in \WEs_\pi \Rightarrow f \in \CEs_\pi\) holds true for all \(f\in\WEs_\pi\).
		Since \(f\) is an entire function, there exists a unique sequence \((a_n)_{n\in\N}\) of real numbers such that 
		\eqref{eq:power-series} holds true. Hence, we can prove the second part of the theorem by showing that
		\((a_n)_{n\in\N}\) is indeed a computable sequence of computable numbers, whenever \(f\) satisfies \(f\in\WEs_\pi\).
		Observe that by \emph{Cauchy's integral formula}, we have
		\begin{align*} 	\frac{a_n}{n!} = \frac{1}{2\pi\iu} \oint_{|z| = 1} \frac{f(z)}{z^{n+1}}\di{z}.
		\end{align*}
		Assume the pair \(((p_{\tm})_{\tm\in\N^2}),\xi)\) satisfies the requirements of Definition \ref{def:WE_pi}. Furthermore, assume without 
		loss of generality that \(\xi(m_1,m_2) = m_1\) holds true for all \(m_1,m_2\in\N\), 
		i.e., \((p_{\tm})_{\tm\in\N^2}\) satisfies
		\begin{align*} \left|f(z) - p_{m_1,m_2}(z)\right| < \frac{1}{2^{m_1}}
		\end{align*}
		for all \(m_1,m_2\in\N, z\in\C\) that satisfy \(|z| \leq m_2\). In the following, for \(m\in\N\), consider the family \((c_k^m)_{k=0}^{K(m)}\)
		of coefficients of \(p_{m,1}\). That is, we have
		\begin{align*} p_{m,1}(z) = \sum_{k=0}^{K(m)} c_k^m \cdot z^k 
		\end{align*}
		for all \(z\in\C,m\in\N\). Then, for \(m\in\N,z\in\C\), we define
		\begin{align*}		r_{n,m} :	&=  \frac{n!}{2\pi\iu} 	\oint_{|z| = 1} \frac{p_{m,1}(z)}{z^{n+1}}\di{z} = c_n^m.														
		\end{align*}	
		We now show that the sequence \((r_{n,m})_{m\in\N}\) converges effectively towards \(a_n\) for all \(n\in\N\), with
		\(|a_n - r_{n,m}| \leq n!\cdot \sfrac{1}{2^m}\) satisfied for all \(m\in\N\). We have
		\begin{align*}
			|a_n - r_{n,m}| 					&\leq n! 	\int_{-\sfrac{1}{2}}^{\sfrac{1}{2}}\left| \frac{f\big(\e^{i2\pi \phi}\big) - p_{m,1}\big(\e^{i2\pi \phi}\big)}{\e^{i2\pi \phi(n+1)}}\right|\di{\phi} \\
																&< 		n! 	\int_{-\sfrac{1}{2}}^{\sfrac{1}{2}} \sfrac{1}{2^m}\di{\phi} \\
																&= \frac{n!}{2^m},
		\end{align*}
		which proves the assertion made above. Hence, \((a_n)_{n\in\N}\) is a computable sequence of computable numbers, making \(f\) and element of \(\CEs_\pi\). 
	\end{proof}

	\begin{corollary}
		If \(f\) satisfies \(f\in\CEs_\pi\), then \(f\) is Markov computable. 
	\end{corollary}
	
	In other words, if \(f\) satisfies \(f\in\CEs_\pi\), there exists a Turing machine \(\TM_f\) that computes the mapping \(x\mapsto f(x)\) for
	\(x\in\Rc\). Hence, for all \(x\in\Rc\), we have \(f(x)\in\Rc\).
	
	While \(\CEs_\pi = \WEs_\pi\) holds true in set-theoretic terms, both \(\CEs_\pi\) and \(\WEs_\pi\) have and individual structure in the sense of
	computability, which we will discuss in the following.
	
	From the proof of Theorem \ref{thm:CEs_equals_WEs}, we conclude that there exists a Turing machine \(\TM\) which, given
	any description \(((p_{\ts})_{\ts\in\N^2},\xi') \describes f\), computes a mapping
	\begin{align*}	((p_{\ts})_{\ts\in\N^2},\xi') \mapsto (a_{n})_{n\in\N},
	\end{align*}
	such that \((a_{n})_{n\in\N} \describes f\) holds true.

	The remainder of this section is dedicated to proving that the converse is not satisfied.
	That is, there does not exist a Turing machine that transforms any description \((a_n)_{n\in\N} \describes f\)
	into a corresponding description \(((p_{\ts})_{\ts\in\N^2},\xi') \describes f\). Note that since \(f\) is bandlimited,
	there exists a number \(L\in\N\) such that
	\begin{align*} \sqrt[n]{|a_n|} \leq L 
	\end{align*}
	holds true for all \(n\in\N\). However, no such \(L\) can be determined 
	algorithmically from the sequence \((a_n)_{n\in\N}\). This observation will play a crucial role in the subsequent lemma. 
	
	\begin{lemma}\label{lem:No_An_to_Pm}
		There exists a sequence \((f_m)_{m\in\N}\) of signals in \(\Es_\pi\) that simultaneously satisfies the following:
		\begin{itemize}
			\item[(a)] For all \(m\in\N\), we have \(f_m \in \CEs_\pi\) (and, consequently, \(f_m \in \WEs_\pi\) as well).
			\item[(b)] There exists a computable double sequence \((a_{n,m})_{n,m\in\N}\) of computable numbers such that
				\((a_{n,m})_{n\in\N} \describes f_m\) in the sense of Definition \ref{def:CEs_pi} holds true for all \(m\in\N\).
			\item[(c)] There does \emph{not} exist a pair \(((p_{m,\ts})_{m\in\N,\ts\in\N^2},\xi'(m,\cdot,\cdot))\) such that
				\(((p_{m,\ts})_{\ts\in\N^2},\xi'(m,\cdot,\cdot)) \describes f_m\) in the sense of Definition \ref{def:WE_pi} holds true for all \(m\in\N\).
		\end{itemize}
	\end{lemma}\begin{proof}
		Let \(A\subset \N\) a be a non-recursive set such that there exists \(g\in\Cas{1}\) that satisfies \(D(g) = A\), i.e.,
		\(A\) is recursively enumerable. Fix any \(n\in\N\) such that \(D(\en_n) = A\) holds true and consider  
		the recursive function \(h:\N^2 \rightarrow \N\) that satisfies 
		\begin{align*} 	h(m,k) = 	\begin{cases}	\min \{l : \Phi(n,m,l) = 1\}, 	&\text{if}~ \Phi(n,m,k) = 1, \\
													k,								&\text{otherwise}.
									\end{cases}
		\end{align*}
		Using this function, define the computable triple sequence \((r_\tm)_{\tm\in\N^3}\) of rational numbers by 
		\begin{align*}	r_\tm :=	\begin{cases} 	2^{-h(m_2,m_3)},	&\text{if}~	0\leq m_1\leq 2^{h(m_2,m_3)} -1, \\
													0,								&\text{otherwise},
									\end{cases}
		\end{align*}
		for all \(\tm\in\N^3\). Next, we introduce the sequence \((a_{\tm})_{\tm\in\N^2}\) of 
		real numbers defined by 
		\begin{align*} 	a_{m_1,m_2}	:=		\begin{cases}	2^{-k},		&\text{if}~\sup_{m_3\in\N} h(m_2,m_3) = k \\
																		&\text{and}~0\leq m_1\leq 2^{k} -1, \\
															0,			&\text{otherwise},
											\end{cases}			
		\end{align*}
		which we will show to be the effective limit of \((r_\tm)_{\tm\in\N^3}\) in \(m_3\) for all \(m_1,m_2\in\N\). 
		In other words, we want to show that \((a_{\tm})_{\tm\in\N^2}\) is a computable double sequence of computable numbers
		such that for some recursive function \(\xi : \N^3 \rightarrow \N\), the pair \(((r_{\tm})_{\tm\in\N^3},\xi)\) is a standard description
		of \((a_{\tm})_{\tm\in\N^2}\). The fact that 
		\begin{align*} \lim_{m_3\to\infty} r_{m_1,m_2,m_3} = a_{m_1,m_2}
		\end{align*}
		holds true for all \(m_3\in\N\) follows directly from
		the definitions of \((r_\tm)_{\tm\in\N^3}\) and \((a_{\tm})_{\tm\in\N^2}\). It thus only remains to show
		that the convergence is effective for all \(m_1,m_2\in\N\). In particular, we want to show that the function 
		\(\xi: (m_1,m_2,M) \mapsto M + 1\) satisfies
		\begin{align}\label{eq:lem:No_An_to_Pm::fooI}
			\bigg| \Big(\lim_{m_3\to\infty}r_{m_1,m_2,m_3}\Big) - r_{m_1,m_2,m_3} \bigg| < \frac{1}{2^M}
		\end{align}
		for all \(m_1,m_2,m_3\in\N\) that satisfy \(m_3 \geq \xi(m_1,m_2,M)\).
		Consider the following case distinction:
		\begin{itemize}
			\item[1.] Fix \(m_1,m_2\in\N\) and assume that there exists \(k\in\N\) such that \(\sup_{m_3\in\N} h(m_2,m_3) = k\) and \(2^k-1 < m_1\) hold true.
				Then, \(r_{m_1,m_2,m_3} = 0\) holds true for all \(m_3\in\N\). Hence, \eqref{eq:lem:No_An_to_Pm::fooI} is satisfied trivially.
			\item[2.] Fix \(m_1,m_2\in\N\) and assume that there exists \(k\in\N\) such that \(\sup_{m_3\in\N} h(m_2,m_3) = k\) and \(2^k-1 \geq m_1\) hold true.
				Then, for all \(m_3\in\Nplus\) that satisfy \(m_3 < \log_2 (m_1 + 1)\), we have
				\begin{align*} \big|a_{m_1,m_2} - r_{m_1,m_2,m_3}\big| = \big| 0 - 2^{-k} \big| \leq 2^{-(m_3-1)}.
				\end{align*}
				For all \(m_3\in\Nplus\) that satisfy \(\log_2 (m_1 + 1) \leq m_3 \leq k\), we have
				\begin{align*} \big|a_{m_1,m_2} - r_{m_1,m_2,m_3}\big| = \big| 2^{-m_3} - 2^{-k} \big| < 2^{-(m_3-1)}.
				\end{align*}
				For all \(m_3\in\Nplus\) that satisfy \(k < m_3\), we have
				\begin{align*} \big|a_{m_1,m_2} - r_{m_1,m_2,m_3}\big| = \big| 2^{-k} - 2^{-k} \big| < 2^{-(m_3-1)}.
				\end{align*}
				Hence, \eqref{eq:lem:No_An_to_Pm::fooI} is satisfied for all \(m_3\in\Nplus\). 
			\item[3.] Fix \(m_1,m_2\in\N\) and assume that there does not exist \(k\in\N\) such that \(\sup_{m_3\in\N} h(m_2,m_3) = k\) holds true.
				Then, for all \(m_3\in\Nplus\) that satisfy \(m_3 < \log_2 (m_1 + 1)\), we have
				\begin{align*} \big|a_{m_1,m_2} - r_{m_1,m_2,m_3}\big| = \big| 0 - 0 \big| < 2^{-(m_3-1)}.
				\end{align*}
				For all \(m_3\in\Nplus\) that satisfy \(\log_2 (m_1 + 1) \leq m_3\), we have
				\begin{align*} \big|a_{m_1,m_2} - r_{m_1,m_2,m_3}\big| = \big| 2^{-m_3} - 0 \big| < 2^{-(m_3-1)}.
				\end{align*}
				Hence, \eqref{eq:lem:No_An_to_Pm::fooI} is satisfied for all \(m_3\in\Nplus\). 
		\end{itemize}
		We conclude that \eqref{eq:lem:No_An_to_Pm::fooI} is satisfied
		for all \(m_1,m_2,m_3\in\N\) that satisfy \(m_3 \geq M = \xi(m_1,m_2,M)\). Hence, the pair \(((r_{\tm})_{\tm\in\N^3},\xi)\) is a standard description
		of the computable double sequence \((a_{\tm})_{\tm\in\N^2}\) of computable numbers. Consider now the sequence \((\ap_{\tm})_{\tm\in\N^2}\)
		defined by 
		\begin{align*} \ap_{m_1,m_2} := m_1!\cdot a_{m_1,m_2}
		\end{align*}
		for all \(m_1,m_2\in\N\). Since \((a_{\tm})_{\tm\in\N^2}\) is a computable double sequence of computable numbers, so is \((\ap_{\tm})_{\tm\in\N^2}\).
		Hence, the sequence \((f_m)_{m\in\N}\) defined by
		\begin{align*}	f_m(z) := \sum_{n=0}^{\infty} \frac{\ap_{n,m}}{n!}z^n = \sum_{n=0}^{\infty} a_{n,m}z^n,	\qquad z\in\C,
		\end{align*}
		is a computable sequence of entire functions in the sense of Definition \ref{def:CEs_pi}, i.e., \((\ap_{\tm})_{\tm\in\N^2} \describes (f_m)_{m\in\N}\). 
		Furthermore, we have \(\bw(f_m) = 0\) for all \(m\in\N\). Thus, \((f_m)_{m\in\N}\) satisfies
		\(f_m\in\CEs_\pi\) for all \(m\in\N\). Defining \(\sup_{l\in\N}h(m,l) =: k(m)\), we have
		\begin{align}	f_m(z) = \sum_{n = 0}^{2^{k(m)}-1} \frac{1}{2^{k(m)}}z^n, \qquad z\in\C
									\label{eq:lem:No_An_to_Pm::fooII}
		\end{align}
		for all \(m\in\N\) that satisfy \(m\in A\), while for all \(m\in\N\) that satisfy \(m\notin A\), we have
		\begin{align*} f_m(z) = 0, \qquad z\in\C. 	
		\end{align*}
		In particular, as follows from \eqref{eq:lem:No_An_to_Pm::fooII} by direct calculation, we obtain that
		\begin{align}	\label{eq:lem:No_An_to_Pm::fooIII} 
									f_m(1) &= \mathds{1}_{A}(m)
		\end{align}
		holds true for all \(m\in\N\).
		We now show that \((f_m)_{m\in\N}\) is not a computable sequence in the sense of Definition \ref{def:WE_pi}. That is,
		there does not exist a pair \(((p_{m,\ts})_{m\in\N,\ts\in\N^2},\xi'(m,\cdot,\cdot))\) consisting of a
		computable triple sequence \((p_{m,\ts})_{m\in\N,\ts\in\N^2}\) of rational polynomials and a recursive function \(\xi:\N^3\rightarrow \N\)
		such that \(((p_{m,\ts})_{m\in\N,\ts\in\N^2},\xi'(m,\cdot,\cdot)) \describes (f_m)_{m\in\N}\) holds true.
		We show this claim by contradiction: assume there exists a pair \(((p_{m,\ts})_{m\in\N,\ts\in\N^2},\xi'(m,\cdot,\cdot))\) such that
		\(((p_{m,\ts})_{m\in\N,\ts\in\N^2},\xi'(m,\cdot,\cdot)) \describes (f_m)_{m\in\N}\) holds true. Since \((p_{m,\ts})_{m\in\N,\ts\in\N^2}\) is a computable sequence
		of polynomials, there exists a Turing machine \(\TM\) that computes the mapping 
		\begin{align*} m \mapsto p_{m,\xi'(m,\sfrac{1}{2},1),1}(1).
		\end{align*}
		Then, for all \(m\in\N\), this Turing machine satisfies 
		\begin{align} 	\label{eq:lem:No_An_to_Pm::fooIV} 
										\left| f_m(1) - \TM(m) \right| < \frac{1}{2}. 
		\end{align}
		Since \(\TM\) furthermore satisfies \(\TM(m) \in \Q\) for all \(m\in\N\), there exists a recursive function \(g:\N\rightarrow \{0,1\}\)
		that satisfies 
		\begin{align*}	g(m) = 	\begin{cases}		1, &\text{if}~ \sfrac{1}{2} < \TM(m),\\
													0, &\text{otherwise}.
								\end{cases}
		\end{align*}
		Incorporating \eqref{eq:lem:No_An_to_Pm::fooIII} and \eqref{eq:lem:No_An_to_Pm::fooIV}, we obtain \(g = \mathds{1}_{A}\), which, since
		\(g\) is a recursive function, is a contradiction to the non-recursivity of \(A\).
	\end{proof}
	
	\begin{theorem}\label{thm:No_An_to_Pm}
		There does \emph{not} exist a Turing machine that computes a mapping 
		\((a_n)_{n\in\N} \mapsto ((p_{\tm})_{\tm\in\N^2},\xi)\), 
		such that \(((p_{\tm})_{\tm\in\N^2},\xi)\describes f\) holds true for every signal \(f\in\CEs_\pi\) 
		and every description \((a_n)_{n\in\N}\describes f\).
	\end{theorem}\begin{proof}
		The claim is a direct consequence of Lemma \ref{lem:No_An_to_Pm} by contradiction: assume there exists a Turing machine \(\TM\) that computes a mapping 
		\((a_n)_{n\in\N} \mapsto ((p_{\tm})_{\tm\in\N^2},\xi)\) for \((a_n)_{n\in\N}\describes f\in\CEs_\pi\),
		such that \(((p_{\tm})_{\tm\in\N^2},\xi)\describes f\) holds true. Consider any pair \(((f_m)_{m\in\N}, (a_{n,m})_{n,m\in\N})\) consisting of 
		a sequence \((f_m)_{m\in\N}\) of signals in \(\Es_\pi\) and a computable double sequence \((a_{n,m})_{n,m\in\N}\) of computable numbers such that
		\(f_m \describes (a_{n,m})_{n\in\N}\) holds true for all \(m\in\N\). Then, setting 
		\begin{align*} ((p_{m,\ts})_{\ts\in\N^2},\xi'(m,\cdot,\cdot)) := \TM((a_{n,m})_{n\in\N})
		\end{align*}
		for all \(m\in\N\) yields a pair \(((p_{m,\ts})_{m\in\N,\ts\in\N^2},\xi'(m,\cdot,\cdot))\) such that
		\(((p_{m,\ts})_{\ts\in\N^2},\xi'(m,\cdot,\cdot))\describes f_m \) holds true for all \(m\in\N\). Hence, if the Turing machine
		\(\TM\) existed, any sequence \((f_m)_{m\in\N}\) of signals in \(\Es_\pi\) that satisfies conditions (a) and (b) of
		Lemma \ref{lem:No_An_to_Pm} would necessarily violate condition (c).
	\end{proof}

	In Section \ref{sec:comp_upper_lower} and Section \ref{sec:semi_decidability}, we will analyze whether \(\bw(f)\) can be computed for \(f \in \CEs_{\pi}\),
	or, if that is not possible, whether we can at least algorithmically determine upper and
	lower bounds for \(\bw(f)\). In Section \ref{sec:arith_complexity}, we will characterize the bandwidth of computable bandlimited signals with respect to the 
	Zheng-Weihrauch hierarchy of real numbers. Furthermore, in Section \ref{sec:OracleComputations}, we will relate the bandwidth of
	computable bandlimited signals with specific properties to a number of common computational oracles.

\section{Computability of Upper and Lower Bounds}
\label{sec:comp_upper_lower}
	Let \(f \in \CEs_{\pi}\). Then we have \(\bw(f) \leq \pi\). As indicated in the introduction,
	it was posed as an open question in \cite{boche21d_accepted} whether it is possible to 
	compute meaningful bounds for the number \(\bw(f)\), given that \(f\) satisfies \(f \in \CEs_{\pi}\).
	Hence, we study the following questions:
	\begin{enumerate}
		\item\label{item:1} Let \(\TMu\) be a Turing machine with \(\TMu \colon \CEs_{\pi} \to \Rc^{+0}\)
			and \(\TMu(f) \geq \bw(f)\) for all \(f\in\CEs_{\pi}\).
			What is the output behavior of this Turing machine?
		\item\label{item:2} Let \(\TMl\) be a Turing machine with \(\TMl \colon \CEs_{\pi} \to \Rc^{+0}\)
			and \(\TMl(f) \leq \bw(f)\) for all \(f\in\CEs_{\pi}\).
			What is the output behavior of this Turing machine?
	\end{enumerate}
	
	Question~\ref{item:1} is concerned with finding Turing machines that compute an upper
	bound for the actual bandwidth of the input signal.
	Similarly, Question~\ref{item:2} is concerned with finding Turing machines that compute a
	lower bound for the actual bandwidth of the input signal.

	In \cite{boche20d} we have constructed a real-valued signal \(f_1 \in \CBs_{\pi}^1\) such
	that \(\bw(f_1) \not\in \Rc\), i.e. 
	the actual bandwidth \(\bw(f_1)\) is not computable. However, for some practical applications,
	it might be sufficient to know a non-trivial lower or upper bound for the number \(\bw(f_1)\).
	In light of this, Questions \ref{item:1} and \ref{item:2} are of high importance. 	
	We show in the following that Turing machines that satisfy the above requirements 
	and yield non-trivial bounds on the number \(\bw(f)\) for at least some feasible signals 
	\(f \in \CEs_{\pi}\) cannot exist, proving the corresponding conjecture  from \cite{boche21d_accepted} correct.

	\begin{theorem}\label{th:turing-upper}
		Let \(\TMu\) be an arbitrary Turing machine such that for all \(f \in \CEs_{\pi}\)
		and all \((a_n)_{n\in\N} \describes f\), \(\TMu\) computes a mapping
		\begin{align*} (a_n)_{n\in\N} \mapsto ((r_n)_{n\in\N},\xi) \describes \TMu(f) \in \Rc
		\end{align*}
		that satisfies \(\TMu(f) \geq \bw(f)\) for all \(f \in \CEs_{\pi}\).
		Then we have \(\TMu(f) = \pi\) for all \(f \in \CEs_{\pi}\).
	\end{theorem}

	Theorem~\ref{th:turing-upper} provides an answer to Question~\ref{item:1}.
	Any Turing machine that computes an upper bound of \(\bw(f)\) for all signals \(f \in \CEs_{\pi}\) is
	necessarily trivial in the sense that it returns the value \(\pi\) for all signals \(f \in \CEs_{\pi}\).
			
	\begin{proof}[Proof of Theorem~\ref{th:turing-upper}]
		We do a proof by contradiction and assume that there exists a Turing machine \(\TMu\) with
		properties as in the theorem and a signal \(f \in \CEs_{\pi}\) such that
		\(\TMu(f) < \pi\). 
		Since \(\TMu(f) \geq \bw(f)\), we also have \(\bw(f) < \pi\).
		For \(\lambda \in [0,1] \cap \Rc\), let
		\begin{align*}
			f_{\lambda}(z)
			=
			(1-\lambda) f(z) + \lambda \frac{\sin(\pi z)}{\pi z} .
		\end{align*}
		We have \(\bw(f_\lambda) = \pi\) for all \(\lambda \in (0,1] \cap \Rc\)
		as well as \(f_\lambda\in\CEs_\pi\) for all \(\lambda \in [0,1] \cap \Rc\),
		the latter being due to the linear structure of \(\CEs_\pi\).
		Furthermore, observe that the mapping \(\lambda \mapsto f_{\lambda}(z)\)
		is computable, i.e., there exists a Turing machine which, for all \(\lambda \in [0,1] \cap \Rc\)
		and all \(((r'_n)_{n\in\N},\xi')\describes\lambda\)
		computes a mapping \(((r'_n)_{n\in\N},\xi') \mapsto (a_n)_{n\in\N}\), such that 
		\((a_n)_{n\in\N}\describes f_\lambda\) is satisfied. Hence, by concatenation, we obtain
		a computable mapping
		\begin{align*}
			((r'_n)_{n\in\N},\xi') \quad \mapsto \quad (a_n)_{n\in\N} \quad \mapsto \quad ((r_n)_{n\in\N},\xi),  
		\end{align*}
			where \(((r_n)_{n\in\N},\xi)\) is a standard description of \(\TMu(f) \in \Rc\), such that
		\begin{align*}
			\lim_{n\to\infty} r'_n = 0 \Leftrightarrow \lim_{n\to\infty} r_n < \pi 
		\end{align*}
		holds true. Since \(\pi\) satisfies \(\pi\in\Rc\), there exists a non-negative, monotonically
		non-decreasing, computable sequence \((s_n)_{n\in\N}\) of rational numbers that satisfies
		\(	\lim_{n\to\infty} s_n = \pi.
		\)
		For all \(n\in \N\) define 
		\begin{align*}
			\overline{r}_n := \min\big\{ r_{\xi(m)} + \sfrac{1}{2^m} : m \leq n\big\} - s_n.
		\end{align*}
		Then, \((\overline{r}_n)_{n\in\N}\) is a monotonically non-increasing, computable sequence
		of rational numbers that satisfies
		\begin{align*}
			\lim_{n\to\infty} r_n < \pi \Leftrightarrow \exists n\in \N: \overline{r}_n < 0. 
		\end{align*}
		Likewise, for \(n\in \N\) define 
		\begin{align*}
			\underline{r}_n := \max\big\{ r'_{\xi'(m)} - \sfrac{1}{2^m} : m \leq n\big\}.
		\end{align*}
		Then, \((\underline{r}_n)_{n\in\N}\) is a monotonically non-decreasing, computable sequence
		of rational numbers that satisfies
		\begin{align*}
			\lim_{n\to\infty} r'_n = 0 \Leftrightarrow \forall n\in \N: \underline{r}_n \leq 0. 
		\end{align*}
		In summary, we observe that both
		\begin{align*}
				\lambda = 0 &\Leftrightarrow \exists n\in \N: \overline{r}_n < 0, \\
				\lambda > 0 &\Leftrightarrow \exists n\in \N: \underline{r}_n > 0 
		\end{align*}
		are satisfied. Define 
		\begin{align*}
			n_0 := \min \{n\in\N : \overline{r}_n < 0 \vee \underline{r}_n > 0\}.
		\end{align*}
		Then, the mapping
		\( ((r'_n)_{n\in\N},\xi') \mapsto n_0
		\)
		is recursive and the Turing machine
		\begin{align*}
			\TM_0(((r'_n)_{n\in\N},\xi')) := 	\begin{cases} 1 &\text{if~} \underline{r}_{n_0} > 0 \\
																											0 &\text{if~} \overline{r}_{n_0} < 0
																				\end{cases}
		\end{align*}
		is well-defined. But \(\TM_0\) satisfies \(\TM_0(((r'_n)_{n\in\N},\xi')) = 0\)
		if and only if \(((r'_n)_{n\in\N},\xi') \describes 0\) holds true, which is a contradiction to
		\cite[Proposition~0, p.~14]{pour-el89_book}.
	\end{proof}

	\begin{theorem}\label{th:turing-lower}
		Let \(\TMl\) be an arbitrary Turing machine such that for all \(f \in \CEs_{\pi}\)
		and all \((a_n)_{n\in\N} \describes f\), \(\TMl\) computes a mapping
		\begin{align*} (a_n)_{n\in\N} \mapsto ((r_n)_{n\in\N},\xi) \describes \TMl(f) \in \Rc
		\end{align*}
		that satisfies \(\TMl(f) \leq \bw(f)\) for all \(f \in \CEs_{\pi}\).
		Then we have \(\TMl(f) = 0\) for all \(f \in \CEs_{\pi}\).
	\end{theorem}
	
	Since zero is the trivial lower bound for \(\bw(f)\), Theorem~\ref{th:turing-lower} shows
	that for \(\CEs_{\pi}\), it is not possible to construct a Turing machine that computes a
	non-trivial lower bound for \(\bw(f)\) for all \(f \in \CEs_{\pi}\). Hence,
	Theorem~\ref{th:turing-lower} provides an answer to Question~\ref{item:2}
	
	\begin{proof}[Proof of Theorem~\ref{th:turing-lower}]
		Assume that there exists a Turing machine \(\TMl\) with properties as in the theorem and a
		signal \(f \in \CEs_{\pi}\) such that \(\TMl(f)>0\). Since \(\TMl(f) \leq \bw(f)\), we also have \(\bw(f)>0\).
		Let
		\begin{align*}
			f(z)
			=
			\sum_{n=0}^{\infty} \frac{a_n}{n!} z^n , \quad z \in \C ,
		\end{align*}
		be the power series representation of \(f\).
		Then the sequence \((a_n)_{n \in \N}\) is a computable sequence of computable numbers,
		and we have
		\begin{align*}
			\limsup_{n \tot \infty} \sqrt[n]{\lvert a_n \rvert}
			\leq
			\pi .
		\end{align*}
		Let
		\begin{align*}
			a_n(m)
			=
			\begin{cases}
				a_n , & 0 \leq n \leq m, \\
				0 , & n > m.
			\end{cases}
		\end{align*}
		The sequence \((a_n(m))_{n \in \N, m \in \N}\) is a computable double sequence of
		computable numbers.

		Let \(\As \subsetneq \N\) be a recursively enumerable non-recursive set, and
		\(\phi_{\As} \colon \N \to \As\) a recursive enumeration of the elements of \(\As\),
		where \(\phi_\As\) is a one-to-one function, i.e., for every element \(m \in \As\) there
		exists exactly one \(\hat{k} \in \N\) with \(\phi_{\As}(\hat{k})=m\).
		Let
		\begin{align*}
			a_n(m,k)
			=
			\begin{cases}
				a_n(\hat{k}), & \text{if \(m\in \{\phi_{\As}(1), \dots, \phi_{\As}(k)\}\)}, \\
				a_n(k) , & \text{otherwise} .
			\end{cases}
		\end{align*}
		Then \((a_n(m,k))_{n \in \N,m \in \N,k \in \N}\) is a computable triple sequence of
		computable numbers.
		For \(n,m \in \N\), we will analyze the behavior of the sequence \((a_n(m,k))_{k \in \N}\)
		next.
		For \(m \in \As\), we have
		\begin{align*}
			\lim_{k \tot \infty} a_n(m,k)
			=
			a_n(\hat{k}) ,
		\end{align*}
		where \(\phi_{\As}(\hat{k})=m\).
		For \(m \not\in \As\), we have
		\begin{align*}
			\lim_{k \tot \infty} a_n(m,k)
			=
			a_n .
		\end{align*}
		For \(n \in \N\), we set
		\begin{align*}
			\hat{a}_n(m)
			=
			\begin{cases}
				a_n(\hat{k}) , & \text{if \(m \in \As\)}, \\
				a_n , &  \text{if \(m \not\in \As\)} ,
			\end{cases}
		\end{align*}
		where \(\phi_{\As}(\hat{k}) = m\).

		Next, we show that there exists a recursive function \(\xi \colon \N^3 \to \N\) such that for
		all \(M \in \N\), \(n \in \N\), and \(m \in \N\), we have
		\begin{align*}
			\lvert \hat{a}_n(m) - a_n(m,k) \rvert
			\leq
			2^{-M}
		\end{align*}
		for all \(k \geq \xi(M,n,m)\).

		Let \(n \in \N\) and \(m \not\in \As\) be arbitrary but fixed.
		Then, for \(k \geq n\), we have
		\begin{align*}
			a_n(m,k)
			=
			a_n(k)
			=
			a_n
			=
			\hat{a}_n(m) .
		\end{align*}
		It follows that for \(M \in \N\) we have
		\begin{align*}
			\lvert \hat{a}_n(m) - a_n(m,k) \rvert
			=
			0
			\leq
			2^{-M} ,
		\end{align*}
		for all \(k \geq \xi(M,m,n) = n\).
		Let \(n \in \N\) and \(m \in \As\) be arbitrary but fixed.
		Then, for \(k \geq n\), we have:
		If \(m \in \{\phi_{\As}(1), \dots, \phi_{\As}(k)\}\), i.e., if
		\(m=\phi_{\As}(\hat{k})\) for \(\hat{k} \leq k\), then we have
		\begin{align*}
			a_n(m,k)
			=
			a_n(\hat{k})
			=
			\hat{a}_n(m) ,
		\end{align*}
		and it follows that
		\begin{align*}
			\lvert \hat{a}_n(m) - a_n(m,k) \rvert
			=
			0
			\leq
			2^{-M}
		\end{align*}
		for all \(M \in \N\).
		If \(m \not\in \{\phi_{\As}(1), \dots, \phi_{\As}(k)\}\), i.e., if \(\hat{k} > k\), then
		we have
		\begin{align*}
			a_n(m,k)
			=
			a_n(k)
			=
			a_n
			=
			a_n(\hat{k})
			=
			\hat{a}_n(m) ,
		\end{align*}
		where we used that \(\hat{k} > k \geq n\) in the second to last equality.
		Hence, for \(M \in \N\) we have
		\begin{align*}
			\lvert \hat{a}_n(m) - a_n(m,k) \rvert
			=
			0
			\leq
			2^{-M}
		\end{align*}
		for all \(k \geq \xi(M,m,n) = n\).

		Combining all partial results, we see that, for all \(M \in \N\), \(m \in \N\), \(n \in \N\),
		and \(k \geq \xi(M,m,n)\) we have
		\begin{align*}
			\lvert \hat{a}_n(m) - a_n(m,k) \rvert
			\leq
			2^{-M} .
		\end{align*}
		This shows that \((\hat{a}_n(m))_{n \in \N, m \in \N}\) is a computable double sequence
		of computable numbers.

		Let
		\begin{align*}
			g_{m}(z)
			=
			\sum_{n=0}^{\infty} \frac{\hat{a}_n(m)}{n!} z^n , \quad z \in \C .
		\end{align*}
		Since
		\begin{align*}
			\limsup_{n \tot \infty} \sqrt[n]{\lvert \hat{a}_n(m) \rvert}
			\leq
			\limsup_{n \tot \infty} \sqrt[n]{\lvert a_n \rvert}
			\leq
			\pi ,
		\end{align*}
		we see that \((g_m)_{m \in \N}\) is a computable sequence of computable signals in
		\(\CEs_{\pi}\).

		We now construct a Turing machine \(\TM\) that consists of two Turing machines.
		The fist Turing machine, which we denote by \(\TM_1 \colon \N \to \{\text{\(\TM\) stops,
			\(\TM\) runs forever}\}\) stops for an input \(m \in \N\) if and only if \(m \in \As\).
		The second Turing machine \(\TM_2\colon \N \to \{\text{\(\TM\) stops, \(\TM\) runs
			forever}\}\) works as follows. 
		For an input \(m \in \N\), it first computes \(g_m \in \CEs_{\pi}\).
		Then it computes \(\TMl(g_m)\).
		We use the fact that there exists a Turing machine \(\TM_{\text{>}}\) that, given an input
		\(\lambda \in \Rc\), stops if and only if \(\lambda > 0\) \cite[Proposition~0,
		p.~14]{pour-el89_book}.
		Our Turing machine \(\TM_2\) now starts \(\TM_{\text{>}}\) with \(\TMl(g_m)\) as an input.
		Hence, \(\TM_2\) will stop if and only if \(\TMl(g_m)>0\).
		For \(m \in \As\), \(g_m\) is a polynomial, for which we have \(\bw(g_m)=0\), according to
		\eqref{eq:actual-bandwidth-limsup-sqrt-an}. 
		Hence, for \(m \in \As\), we have \(\TMl(g_m)=0\).
		Consequently, \(\TM_2\) stops if and only if \(m \not\in \As\).
		Thus, \(\TM\), which runs \(\TM_1\) and \(\TM_2\) in parallel, is a Turing machine that, for
		any input \(m \in \N\), can decide whether \(m \in \As\) or \(m \not\in \As\).
		This implies that the set \(\As\) is recursive, which is a contradiction.
		Hence, \(\TMl\) cannot exist.
	\end{proof}

\section{Semi-Decidability}
\label{sec:semi_decidability}
	On a related note, one may ask whether, for a given value \(\sigma \in (0,\pi) \cap \Rc\), it is
	possible to algorithmically detect if a signal \(f\in \CEs_{\pi}\) satisfies the bandwidth
	condition \(\bw(f) > \sigma\). 
	Ideally, we would like to have a Turing machine
	\(\TM_\sigma \colon \CEs_{\pi} \to \{\text{true, false}\}\) that satisfies
	\(\TM_\sigma(f) = \text{true}\) if and only if \(\bw(f) > \sigma\) holds true. 
	However, such a Turing machine cannot exist: relabeling the output-values of
	\(\TM_\sigma\) by "\(\sigma\)" for "true" and "\(0\)" for "false", we immediately obtain a
	Turing machine that computes a non-trivial lower bound on \(\bw(f)\) for some
	\(f \in \CEs_{\pi}\), thereby contradicting Theorem~\ref{th:turing-lower}.

	As discussed in the introduction, another open question was posed in \cite{boche21d_accepted}, 
	concerning a weakening of the above problem:
	Does there exist an algorithm \(\TM_{\text{BW}>\sigma}\) that, upon being presented with
	an input \(f\in\CEs_{\pi}\), halts its computation in a finite number of steps if
	\(\bw(f) > \sigma\) holds true, and computes forever otherwise? 

	To answer this question we introduce the concept of semi-decidability.
	We call a set \(\Ms \subset \CEs_{\pi}\) semi-decidable if there exists a Turing machine
	\(\TM \colon \CEs_{\pi} \to \{\text{\(\TM\) stops, \(\TM\) runs forever}\}\) that, given an
	input \(f \in \CEs_{\pi}\), stops if and only if \(f \in \Ms\). 
	Consequently, the question about the existence of \(\TM_{\text{BW}>\sigma}\) for some
	value \(\sigma \in (0,\pi)\) is equivalent to the question about the semi-decidability of
	the set \(\left\{f \in \CEs_{\pi} \colon \bw(f) > \sigma\right\}\).

	In \cite{boche20d}, this question was studied for signals in \(\Bs_{\pi}^1\), and it was
	proved that the set \(\{ f \in \Bs_{\pi}^1 \colon \bw(f) > \sigma \}\) is semi-decidable.
	Next, we show that \(\left\{f \in \CEs_{\pi} \colon \bw(f) > \sigma\right\}\) is not
	semi-decidable.
	
	\begin{theorem}\label{th:semi-decidable-lower} 
		For all \(\sigma \in (0,\pi) \cap \Rc\), the set
		\begin{align}\label{eq:th:semi-decidable-lower:set}
			\left\{
				f \in \CEs_{\pi} \colon \bw(f) > \sigma
			\right\}
		\end{align}
		is not semi-decidable.
	\end{theorem}\begin{proof}
		  We do a proof by contradiction and assume that there exists 
		  $\sigma \in (0,\pi) \cap \Rc$ and a Turing machine $\TM_{\text{BW}>\sigma}$ 
		  that accepts exactly the set \(\CEs_{\pi}^{>}(\sigma)\).
		  Let $f_* \in \CEs_{\pi}$ with $\bw(f_*)>\sigma$ be arbitrary.
		  Analogous to the proof of Theorem \ref{th:turing-lower}, we construct a computable sequence
		  \((g_m)_{m\in\Nplus}\) of functions in \(\CEs_\pi\) that satisfies
		  \begin{align*}
			\bw(g_m) 	=	\begin{cases}	\bw(f_*),	&\text{if}~ m\in\Nplus\setminus\As, \\	
											0, 			&\text{otherwise},
							\end{cases} 
		  \end{align*}
		  for some recursively enumerable, non-recursive set \(\As \subsetneq \Nplus\).
		  Then, the Turing machine $\TM_{\text{BW}>\sigma}$ provides an algorithm
		  that semi-decides the set \(\Nplus\setminus \As\), which contradicts the assumption.
\end{proof}

	In the broadest sense, the Theorems \ref{th:turing-upper}, \ref{th:turing-lower} and
	\ref{th:semi-decidable-lower} yield a fundamental limit to the capabilities of
	computer-aided design. 
	Algorithms for automated system design commonly feature an \emph{exit flag} functionality,
	that, upon termination of the algorithm, indicates whether the computation was successful
	or not.

	In our context, a hypothetical computer-aided design tool might want to sample a signal
	with a given sampling rate \(\sigma/\pi\).
	It would return the sampled signal together with a boolean exit flag that signals whether
	we had a problematic input signal, having a bandwidth that exceeds \(\sigma\).
	However, as Theorem~\ref{th:semi-decidable-lower} shows, such an exit flag functionality
	cannot be implemented for signals in \(\CEs_{\pi}\).

\section{Structural Properties of the Bandwidth of Computable Signals}
\label{sec:arith_complexity}
	The bandwidth of signals in \(\CBs_{\pi}^1\) is not always a computable number, as was shown in
	\cite{boche21d_accepted}. Hence, the question arises whether it is possible to characterize the degree of 
	uncomputability of the elements of \(\bw[\CBs_\pi^p]\) for \( 1 \leq p \leq +\infty\) and, more generally, \(\bw[\CEs_\pi]\).
	That is, we want to classify the elements of \(\bw[\CBs_{\pi}^p]\) for \( 1 \leq p \leq +\infty\) and \(\bw[\CEs_\pi]\)
	with respect to the arithmetical hierarchy of real numbers. 
	
	The majority of mathematical models in science and engineering are, in the final analysis, supposed to yield
	quantitative results. Despite not having received much attention in the past, a classification of such models regarding
	their arithmetic complexity is highly desirable. Not only does it yield essential insights concerning the mathematical
	structure of the model itself, but also answers whether, given a certain type of computation machine, we can expect the model
	to produce meaningful quantitative results in the first place. To the best of our knowledge, the present work is the first one 
	to yield such a characterization for a problem from engineering.
	
	As mentioned above, the bandwidth of signals in \(\CBs_{\pi}^1\) is generally not a computable number. The proof presented
	in \cite{boche21d_accepted} relies on the construction of a signal \(f \in \CBs_{\pi}^1\) which exhibits this property.
	Implicitly, the following assertion was proven as well (c.f. \cite[Theorem~2, p.~7]{boche21d_accepted}):
	
	\begin{lemma}\label{lem:CBsSignal}
		Let \((r_m)_{m\in\N}\) be a bounded, monotonically non-decreasing, computable sequence of positive rational numbers. 
		Then, there exists a signal \(f\in\CBs_{\sigma}^{1}\) that satisfies
		\begin{align}\label{eq:lem:CBsSignal::I} \sigma := \bw(f) = \lim_{m\to\infty} r_m.
		\end{align}
	\end{lemma}\begin{proof}
		Consider the indicator function \(\ind_{[-\sfrac{1}{2},\sfrac{1}{2}]} : \R \mapsto \{0,1\}\) of the
		real interval \([-\sfrac{1}{2},\sfrac{1}{2}]\) in the frequency domain and set \(\Fta := \ind_{[-\sfrac{1}{2},\sfrac{1}{2}]}\).
		Then, we have
		\begin{align*} \fa(t) = \frac{ \sin\left(\sfrac{t}{2}\right)}{\pi t}.
		\end{align*}
		Consider furthermore the convolution \(\Ftb := \Fta*\Fta\) of \(\Fta\) with itself. We have
		\begin{align*}
			\Ftb(\omega) 	&= 	\begin{cases}		1 - |\omega|,					&\quad\text{if}~ |\omega| \leq 1, \\
													0,								&\quad\text{otherwise},
														\end{cases} \\
			\fb(t)			&=	\left( \frac{ \sin\left(\sfrac{t}{2}\right)}{\pi t}\right)^2.										
		\end{align*}
		Hence, \(\fb\) satisfies \(\fb(t) \geq 0\) for all \(t\in \R\) and we have
		\begin{align*}	\|\fb\|_{L^1} = \int_{-\infty}^{+\infty} |\fb(t)| \di{t} = \int_{-\infty}^{+\infty} \fb(t) \di{t} = \Ftb(0) = 1.		
		\end{align*}
		We conclude that \(\fb\) satisfies \(\fb \in \Bs_{1}^1\). It can be shown that \(\fb\) is also computable in 
		\(\Bs_{1}^1\) \cite[Appendix~E, p.~16]{boche21d_accepted}, i.e., \(\fb \in \CBs_{1}^1\).
		We will now consider signals \(\fc_m\) that emerge from \(\fb\) by shortening and shifting in the frequency
		domain. For \(m\in\N\), define
		\begin{align*}		\Ftc_m(\omega):		&= \Ftb\left(\frac{2}{r_{m+1} - r_{m}}\left(\omega - \frac{r_{m+1} + r_{m}}{2}\right)\right), \\
							\fc_m(t)    		&=  \e^{\iu \frac{r_{m+1} + r_{m}}{2}t} \left(\frac{r_{m+1} - r_{m}}{2} \fb\left(\frac{(r_{m+1} - r_{m})t}{2}\right)\right),
		\end{align*}
		whenever \(r_m < r_{m+1}\) holds true. If \(r_m = r_{m+1}\) holds true instead, set \(\Ftc_m = \fc_m = \bm{0}\) (here, \(\bm{0}\) refers to the trivial signal in \(L_1\)).
		We observe that \(\fc_m\) satisfies 
		\begin{align*}	\|\fc_m\|_{L^1} = \|\fb_m\|_{L^1} = 1
		\end{align*}
		for all \(m\in\N\) that satisfy \(r_m < r_{m+1}\) and \(\Ftc_m\) satisfies
		\begin{align*}	\suppess \Ftc_m = [r_m, r_{m+1}] 
		\end{align*}
		for all \(m\in\N\) that satisfy \(r_m < r_{m+1}\). Hence, since \(r_m\) is a rational number for all \(m\in\N\), we have
		\(\fc_m \in \CBs_{r_m+1}^1\) and \(\bw(\fc_m) = r_{m+1} \) for all \(m\in\N\) that satisfy \(r_m < r_{m+1}\). We now define
		\begin{align*} f_k(t) :&= \sum_{m=1}^{k} \frac{1}{m^2}\cdot \fc_m(t),	\\
									f(t) :	&= \lim_{k\to\infty} f_k(t),
		\end{align*}
		and, in the following, show that \(f\) is indeed a signal \(f\in\CBs_{\sigma}^{1}\) that satisfies \eqref{eq:lem:CBsSignal::I}.
		Define \(\Ms := \{m\in\N: r_m < r_{m+1}\}\). First, observe that
		\begin{align*} \|f\|_{L^1} \leq \sum_{m=1}^{\infty} \left\| \frac{1}{m^2}\cdot \fc_m	\right\|_{L^1}
															= \sum_{m\in\Ms} \frac{1}{m^2} \leq \frac{\pi^2}{6}
		\end{align*}
		is satisfied. Since \(\Ftc_m(\omega) \geq 0\) holds true for all \(\omega \in \R \), we have
		\begin{align*} \bw(f) &= \sup \bigcup_{m\in\Ms} \suppess \Ftc_m \\
													&= \sup \bigcup_{m\in\Ms} [r_m, r_{m+1}] \\ 
													&= \lim_{m\to\infty} r_{m+1} 
		\end{align*}
		It remains to show that \(f\) is computable in \(\CBs_{\sigma}^1\). Since \(\fc_m\) is computable in \(\CBs_{\sigma}^1\)
		for all \(m\in\N\), so is \(f_k\) for all \(k\in\N\). Hence, it is sufficient to prove that \(f_k\) converges effectively 
		towards \(f\) for \(k\to\infty\). We have
		\begin{align*} 	\left\| f - f_k	\right\|_{L^1} 
										&\leq \sum_{m=k+1}^{\infty} \frac{1}{m^2}\cdot \left\| \fc_m	\right\|_{L^1} \\
										&\leq \sum_{m=k+1}^{\infty} \frac{1}{m^2} \\
										&= \frac{\pi^2}{6} - \sum_{m=1}^{k} \frac{1}{m^2},
		\end{align*}
		which is a computable sequence of computable numbers in \(k\) that converges monotonically decreasingly towards \(0\). 
		Hence, the convergence of \(f_k\) towards \(f\) is effective and we have \(f\in\CBs_{\sigma}^{1}\) with \(f\)
		satisfying \eqref{eq:lem:CBsSignal::I}.
	\end{proof}

	For the full characterization of \(\bw[\CBs_{\pi}^p]\) for \(1\leq p\leq +\infty\) we first need another lemma. Again, the statement was proven in 
	a related form in \cite[Theorem~5, p.~10]{boche21d_accepted}: for all \(\sigma \in (0,\pi)\cap\Rc\) and all \(1\leq p\leq +\infty\), the set
	\(\CBs_\pi^p\setminus\CBs_\sigma^p\) is semi-decidable with respect to \(\CBs_\pi^p\). That is, for all \(\sigma \in (0,\pi)\cap\Rc\) and all \(1\leq p\leq +\infty\),
	there exists a Turing machine \(\TM\) that computes a partial mapping \(f \mapsto m\in\N\), where \(D(\TM)\) equals \(\CBs_\pi^p\setminus\CBs_\sigma^p\).
	Such a Turing machine exists for all \(\sigma \in (0,\pi)\cap\Rc\) and all \(1\leq p\leq +\infty\), and takes a description of \(f \in \CBs_{\pi}^p\).
	In order to establish the hierarchical characterization of \(\bw[\CBs_{\pi}^p]\), we require the analogous statement for input \(\sigma \in (0,\pi)\cap\Rc\) instead
	of \(f \in \CBs_\pi^p\). That is, we want to show that for all \(f \in \CBs_{\pi}^p\) and all \(1\leq p\leq +\infty\), there exists a Turing machine
	that semi-decides the set \(\{\sigma \in (0,\pi) \cap \Rc: \sigma < \bw(f)\}\) for input \(\sigma \in (0,\pi)\cap\Rc\).

	\begin{lemma}\label{lem:CBs_semi_decidable}
		For all \(f \in \CBs_{\pi}^p\) and all \(1\leq p\leq +\infty\), the set \(\{\sigma \in (0,\pi) \cap \Rc: \sigma < \bw(f)\}\)
		is semi-decidable with respect to the set \((0,\pi) \cap \Rc\). 
	\end{lemma}\begin{proof}
		If \(f\) satisfies \(f \in \CBs_{\pi}^p\) for some \(1 <  p\leq +\infty\), then there exists a signal \(f' \in \CBs_{\pi}^1\) 
		that satisfies \(\bw(f) = \bw(f')\) \cite[Proof of Theorem~5, p.~10]{boche21d_accepted}. Hence, without loss of generality,
		we can assume that \(f \in \CBs_{\pi}^1\) holds true. Since \(f\) satisfies \(f \in \CBs_{\pi}^1\), the Fourier transform \(\Ft{f}\) of 
		\(f\) is Turing computable on the interval \((-\pi,\pi]\). In particular, this implies the existence of a Turing machine \(\TM\) that
		computes a mapping 
		\begin{align*} \sigma \mapsto ((r_m)_{m\in\N},\xi)
		\end{align*}
		for \(\sigma \in (0,\pi) \cap \Rc\), where \(((r_m)_{m\in\N},\xi)\) is a standard description of the number
		\begin{align*}	x(\sigma) := \max \big\{\Ft{f}(\omega) : \omega \in [-\pi, -\sigma) \cup (\sigma,\pi]\big\} \in \Rc.
		\end{align*}
		Then, \(x(\sigma)\) satisfies \(x(\sigma) > 0 \) if and only if \(\sigma < \bw(f)\) holds true.
		Consider the computable sequence \((\rp_m)_{m\in\N}\) of rational numbers that satisfies
		\begin{align*} \rp_m := r_m - 2^{-\xi(m)}.
		\end{align*}
		We have \(\lim_{m\to\infty} \rp_m = x(\sigma)\) as well as \(\rp_m \leq x(\sigma)\) for all \(m\in\N\). Hence, there
		exists an \(m\in\N\) such that \(\rp_m > 0\) holds true if and only if \(\sigma\) satisfies \(\sigma < \bw(f)\).
		Finally, we conclude the existence a Turing machine \(\TM'\) that computes the mapping
		\begin{align*} \sigma \mapsto \min \{m\in\N : \rp_m > 0\}
		\end{align*}
		for input \(\sigma \in (0,\pi) \cap \Rc\). The Turing machine \(\TM'\) then satisfies \(D(\TM') = \{\sigma \in (0,\pi)\cap\Rc : \sigma < \bw(f)\}\), 
		which is the required property. 
	\end{proof}
	
	We will now use the above results to prove the full characterization of \(\bw[\CBs_{\pi}^p]\) for \(1\leq p\leq +\infty\) with
	respect to the arithmetical hierarchy of real numbers.
	\begin{theorem}\label{thm:CBs_Sigma1}\mbox{}
		\begin{itemize}
			\item[1.] If \(f\) satisfies \(f \in \CBs_{\pi}^p\) for \(1\leq p\leq +\infty\), then \(\bw(f)\) satisfies
				\(\bw(f) \in \Sigma_1\).
			\item[2.] If \(x\) satisfies \(x \in \Sigma_1 \cap [0,\pi]\) then there exists \(f \in \CBs_{\pi}^1\) such that
				\(\bw(f) = x\) holds true.
		\end{itemize}
		Hence, the set \(\bw[\CBs_{\pi}^p]\) coincides with the set \(\Sigma_1 \cap [0,\pi] \) for all \(1\leq p\leq +\infty\).
	\end{theorem}\begin{proof}\mbox{}
		\begin{itemize}
			\item[1.] From Lemma \ref{lem:CBs_semi_decidable}, we know that there exists a Turing machine 
				\(\TM\) that computes a mapping \(\sigma \mapsto m \in\N\) for input \(\sigma \in (0,\pi) \cap \Rc\), such that 
				\begin{align*}	D(\TM) = \{\sigma \in (0,\pi) \cap \Rc : \sigma < \bw(f)\}
				\end{align*}
				holds true. In the following, we consider for all \(\tm\in\N^3\) the rational number
				\begin{align*}	q(\tm) := (-1)^{m_3}\frac{m_1}{m_2}.
				\end{align*}
				Since \(\Q\) is a subset of \(\Rc\), Lemma \ref{lem:CBs_semi_decidable} implies the existence of a 
				recursive function \(g'' : \N^3 \rightarrow \N\), such that
				\begin{align*} 	&D(g'') \cap \left\{\tm\in\N^3 : q(\tm) \in (0,\pi) \right\} \\
												&\qquad\qquad= \left\{\tm\in\N^3 : q(\tm) \in D(\TM)\right\} 
				\end{align*}
				holds true. Since both \(0\) and \(\pi\) are computable numbers, the set \((0,\pi)\cap \Q\) is recursive in the following sense:
				there exists a total recursive mapping \(g' : \N \rightarrow \N\) such that
				\begin{align*} q(\amalg_3g'(n)) \in (0,\pi)
				\end{align*}
				holds true for all \(n\in\N\) and for all \(\sigma\in (0,\pi)\cap \Q\), there exists an \(n\in\N\) such that
				\begin{align*} q(\amalg_3g'(n)) = \sigma
				\end{align*}
				holds true. Define \(g(n) := g''(\amalg_3g'(n))\) for all \(n\in\N\). Then, \(g\) is a recursive function that satisfies
				\begin{align*} 	D(g) 	&= \{n\in\N : q(\amalg_3g'(n)) \in (0,\bw(f))\},\\
												g[\N]	&= (0,\bw(f)) \cap \Q.
				\end{align*}
				We have \(\sup (0,\bw(f)) \cap \Q = \sup\{ q \in \Q : q < \bw(f)\}\), where
				\(\{ q \in \Q : q < \bw(f)\}\) is the Dedekind cut of \(\bw(f)\). Hence, \(g\) satisfies
				\(\sup g[\N] = \bw(f)\). Now consider \(k\in\N\) such that \(g = \en_k\) and define the set
				\begin{align*} \Qs(l) := \Big\{q(\amalg_3g'(\varpi_1(m))):~ &m \in \{1,\ldots, l\}, \\
																																		&\Psi(k\circ\amalg_2(m)) = 1\Big\}.
				\end{align*}
				Then, we have \(\Qs(l) \subseteq \Qs(l + 1)\) for all \(l\in\N\) as well as  
				\begin{align*} \bigcup_{l=1}^{\infty} \Qs(l) = g[\N] = (0,\bw(f)) \cap \Q,
				\end{align*}
				and the monotonically non-decreasing, computable sequence \((r_m)_{m\in\N}\) of rational numbers, defined by
				\begin{align*} 	r_m :=	\begin{cases}	\max \Qs(m), &\quad\text{if}~\Qs(m)\neq\emptyset,\\
																							0,						&\quad\text{otherwise},
																\end{cases}
				\end{align*}
				for all \(m\in\N\), satisfies \(\sup_{m\in\N} r_m = \bw(f)\).
			\item[2.] If \(x\) satisfies \(x \in \Sigma_1\), then, per definition, there exists 
				a computable sequence \((\rp_m)_{m\in\N}\) of rational numbers that satisfies 
				\begin{align*} \sup_{m\in\N}(\rp_m) = x.
				\end{align*}
				Now, for all \(m\in\N\), define
				\begin{align*} r_m := \max \big(\{r_k : k \leq m\} \cup\{0\}\big).
				\end{align*}
				Then, \((r_m)_{m\in\N}\) is a non-negative, computable sequence of rational numbers
				that converges monotonically non-decreasingly towards \(x\). Hence, the claim follows from Lemma \ref{lem:CBsSignal}.
		\end{itemize}
	\end{proof}
	
	As Theorem \ref{thm:CBs_Sigma1} shows, the computability properties of \(\bw[\CBs_\pi^p]\) do (for \(p \geq 1\)) not depend
	on the actual value of \(p\). In other words, the computability behavior of signals in \(\CBs_{\pi}^{\infty}\) 
	is not getting any more benevolent if the requirements on the signal decay for \(t \to \infty\) in the time domain are
	strengthened. However, a transition between hierarchical levels does occur when the class of allowed signals is expanded 
	from \(\CBs_{\pi}^{\infty}\) to \(\CEs_\pi\), as we will see in the following.
	
	\begin{theorem}\label{thm:CEs_Pi2}\mbox{}
		\begin{itemize}
			\item[1.] If \(f\) satisfies \(f \in \CEs_{\pi}\), then \(\bw(f)\) satisfies \(\bw(f) \in \Pi_2\).
			\item[2.] If \(x\) satisfies \(x \in \Pi_2 \cap [0,\pi]\) then there exists \(f \in \CEs_{\pi}\) such that
				\(\bw(f) = x\) holds true.
		\end{itemize}
		Hence, the set \(\bw[\CEs_{\pi}]\) coincides with the set \(\Pi_2 \cap [0,\pi] \).
	\end{theorem}\begin{proof}\mbox{}
		\begin{itemize}
			\item[1.] By \eqref{eq:actual-bandwidth-limsup-sqrt-an}, there exists a computable sequence \((a_m)_{m\in\N}\) of 
				computable numbers such that \(\bw(f) = \limsup_{m \tot \infty} \sqrt[m]{\lvert a_m \rvert}\) holds true.
				Define \(b_m := \sqrt[m]{\lvert a_m \rvert}\) for all \(n\in\N\). Then, \((b_m)_{m\in\N}\) is a computable sequence
				of computable numbers and 
				\begin{align*}
					\bw(f) = \limsup_{m \tot \infty} (b_m) 
				\end{align*}
				holds true. In the following, we employ a technique applied by \citeauthor{zw01}. 
				By definition, we have
				\begin{align*}
					&\limsup_{m \tot \infty} (b_m) \\
					&\quad= \lim_{m_1\to\infty}\big(\sup \big\{b_{m_2} : m_2\in\N, m_1 \leq m_2\big\}\big) \\
					&\quad= \lim_{m_1\to\infty}\big(\sup \big\{b_{m_1+m_2} : m_2\in\N\}\big)
				\end{align*}
				Observe that \(\big\{b_{m_1+1+m_2} : m_2\in\N\} \subseteq \big\{b_{m_1+m_2} : m_2\in\N\}\) is satisfied for all \(m_1\in\N\), and hence
				\(\sup \big\{b_{m_1+m_2} : m_2\in\N\}\) is monotonically non-increasing in \(m_2\). Thus,
				\begin{align*}
					&\limsup_{m \tot \infty} (b_m) \\ 		
					&\quad= \lim_{m_1\to\infty}\big(\sup \big\{b_{m_1+m_2} : m_2\in\N\}\big) \vphantom{\lim_{m_1\to\infty}}\\
					&\quad= \inf_{m_1\in\N}\big(\sup \big\{b_{m_1+m_2} : m_2\in\N\}\big) \vphantom{\lim_{m_1\to\infty}}\\
					&\quad= \inf_{m_1\in\N}\sup_{m_2\in\N} \big(b_{m_1+m_2}\big) \vphantom{\lim_{m_1\to\infty}}
				\end{align*} 
				holds true. For all \(m_1,m_2\in\N\), define \(\bp_{m_1,m_2} := b_{m_1+m_2}\). Then, \((\bp_{\tm})_{\tm\in\N^2}\) is a computable
				double sequence of computable numbers that satisfies
				\begin{align*}
					\bw(f) = \inf_{m_1\in\N}\sup_{m_2\in\N} \big(\bp_{\tm}\big).
				\end{align*}
				That is, \((\bp_{\tm})_{\tm\in\N^2}\) is a second order upper ZW description of \(\bw(f)\). 
				The claim then follows by Lemma \ref{lem:LowerZWImpliesPi2}.
			\item[2.] Consider a computable sequence \((r_{\tm})_{\tm\in\N^2}\) of rational numbers that satisfies \(\inf_{m_1\in\N}\sup_{m_2\in\N}(r_{\tm}) = x\).
				Without loss of generality, we can assume that \(r_{\tm}\) is non-negative for all \(\tm\in\N^2\).
				Following Lemma \ref{lem:ZhengWeihrauchI}, there exists a Turing machine that computes a mapping \((r_{\tm})_{\tm\in\N^2} \mapsto (\rp_{m})_{m\in\N}\) 
				such that \((\rp_{m})_{m\in\N}\) is a non-negative computable sequence of rational numbers that satisfies 
				\(\limsup_{m\to\infty} (\rp_{m}) = x\). We define \(a_m := (\rp_m)^m\) for all \(m\in\N\). Then \((a_m)_{m\in\N}\) is a computable sequence
				of rational numbers, and hence a computable sequence of computable numbers as well. By \eqref{eq:actual-bandwidth-limsup-sqrt-an}, the signal
				\begin{align*}
					f(z) = \sum_{m=0}^{\infty} \frac{a_m}{m!} z^m = \sum_{m=0}^{\infty} \frac{(\rp_m z)^m}{m!} , \quad z \in \C ,
				\end{align*}
				then satisfies \(\bw(f) = x\).
		\end{itemize}
	\end{proof}

	\begin{remark}
		Observe that the proof of the second statement of Theorem \ref{thm:CEs_Pi2} is constructive and all involved operations can,
		in principle, be computed by a Turing machine. Hence if we have a hypothetical \emph{Taylor signal generator} available, i.e.,
		a machine that maps a non-negative, \(\pi\)-bounded, computable sequence of rational numbers \((\rp_m)_{m\in\N}\) to
		the \enquote{analog} signal \(\sum_{m=1}^\infty \sfrac{(\rp_m z)^m}{m!} = f(z)\), we can in principle build an apparatus that
		receives a second order upper ZW description of a number \(x\) as an input and returns an analog signal 
		\(f\) that satisfies \(\bw(f) = x\).
	\end{remark}
	
	For \(1 \leq p \leq +\infty\) and \(f\in \CBs_\pi^p\), there exists a monotonically non-decreasing, 
	computable sequence of rational numbers \((r_n)_{n\in\N}\), such that \(\lim_{n\to\infty} r_n = \bw(f)\) is satisfied. 
	As Theorem \ref{thm:CEs_Pi2} shows, this does not hold true for signals \(f\in\CEs_\pi\), even if the monotonicity requirement is made void.
	Assume for some signal \(f\in\CEs_\pi\), there exists a computable sequence \((r_n)_{n\in\N}\) of rational numbers such that
	\(\lim_{n\to\infty} r_n = \bw(f)\) holds true. For all \(m_1,m_2\in\N\), define 
	\begin{align*}	\rp_{m_1,m_2} := r_{m_1+m_2}.
	\end{align*}
	Then, \((\rp_{m_1,m_2})_{m_1,m_2\in\N}\) is a computable double sequence of rational numbers that satisfies
	\begin{align*}	\inf_{m_1\in\N}\sup_{m_2\in\N} (\rp_{m_1,m_2}) = \sup_{m_1\in\N}\inf_{m_2\in\N} (\rp_{m_1,m_2}) = \bw(f).
	\end{align*}
	Hence, we have \(\bw(f) \in \Pi_2 \cap \Sigma_2 = \Delta_2\). But, since \(\bw[\CEs_\pi] = \Pi_2 \cap [0,\pi]\) and
	\((\Pi_2 \cap [0,\pi])\setminus \Sigma_2 \neq \emptyset\) hold true, there exist signals \(f\in \CEs_\pi\) that satisfy
	\(\bw(f)\notin \Delta_2\). Hence, for these signals, there cannot exist computable sequences \((r_n)_{n\in\N}\) of rational numbers
	that satisfy \(\lim_{n\to\infty} r_n = \bw(f)\).

\section{Bandwidth Computation and Oracle Machines}
\label{sec:OracleComputations}
	In this section, we consider bandlimited signals in the context of oracle computations.
	A machine \(\OM\) that computes a function \(g \in \Ca(A)\) for an arbitrary, non-recursive
	set \(A \subsetneq \N\) is referred to as an \emph{oracle machine}. In particular, we refer to
	the oracle machine 
	\begin{align*}	\OM_{\mathrm{T}} : n \mapsto \mathds{1}_{\At}(n)
	\end{align*}
	as \emph{totality oracle}, and to the members
	\begin{align*}	\OM_{\mathrm{H},n} : m \mapsto \mathds{1}_{D(\en_n)}(m)
	\end{align*}
	of the family \((\OM_{\mathrm{H},n})_{n\in\N}\) as \emph{halting oracles}.

	The properties of oracle computations and their relation to non-recursive sets have been intensively studied in the relevant literature;
	for a comprehensive introduction, we refer to \cite{soare87_book}. 
	Hence, any relation to a real world problem from science or engineering is interesting for the field of theoretical computer science.
	Furthermore, it yields another perspective on the \emph{degree} of uncomputability present in the problem.
	
	In the following, we want to relate the problem of computing the bandwidth of a bandlimited signal to these oracles. 
	Therefore, we consider two hypothetical devices that operate on the set \(\CEs_\pi\):
	\begin{itemize}		\item 	The device \(\Ob\) evaluates the mapping
														\begin{align*}		f	\mapsto	\begin{cases}	1, &\text{if}~\bw(f) < \sfrac{1}{2}, \\
																																			0,	&\text{otherwise}.
																												\end{cases}
														\end{align*}
														In other words, the device \(\Ob\) decides, for a signal \(f \in \CEs_\pi\),
														whether \(\bw(f)\) is smaller than \(\sfrac{1}{2}\) or not. Hence, \(\Ob\) yields a non-trivial bound on \(\bw(f)\).
										\item		For \(x \in (0,4]\), we denote by \(A[x] \subseteq \N\) the unique countably infinite set that satisfies
														\begin{align*}	x = 4\cdot\sum_{m\in A[x]}\frac{1}{2^m}.
														\end{align*}
														Furthermore, we set \(A_n[x] := \{ m \in A[x] : m \leq n+2\}\). The device \(\Oa\) then evaluates the mapping
														\begin{align*}	(f,n) \mapsto x[A_n[\bw(f)]].	
														\end{align*}
														Hence, the device \(\Oa\) yields an approximation of the number \(\bw(f)\) which is accurate up to \(n\)-binary places.
	\end{itemize}
	Additionally, we consider a device \(\SG\) which, for \((a_n)_{n\in\N} \describes f\), evaluates the mapping \((a_n)_{n\in\N} \mapsto f\). That is,
	\(\SG\) maps the description \((a_n)_{n\in\N}\) of \(f\) to the actual signal \(f \in \CEs_\pi\) according to \eqref{eq:power-series}.
	
	In the following, we will investigate the computational strength of the oracle machines \(\Ob\) and \(\Oa\). Given the fact 
	that \(\bw[\CEs_\pi] = \Pi_2 \cap [0,\pi]\) and \(\At \in \Pi_2^0\) hold true, it may not be surprising that \(\Ob\) and \(\Oa\) are at 
	least as strong as the totality oracle \(\OM_{\mathrm{T}}\). 

	\begin{theorem}	There exists a Turing machine \(\TM\) which, for \(n\in\N\), computes a mapping \(n \mapsto (a_n)_{n\in\N}\), such that
									\begin{align*}	\Ob(\SG(\TM(n))) = \OM_{\mathrm{T}}(n)
									\end{align*}
									holds true for all \(n\in\N\).
	\end{theorem}\begin{proof}
									Consider the runtime function \(\Psi : \N^3 \rightarrow \{0,1\}\). For \(n\in\N\) fixed, 
									define the computable double sequence \((\rpp_{n,\tm})_{\tm\in\N^2}\) of rational numbers by setting
									\begin{align*}	\rpp_{n,\tm} := \Psi(n,m_1,m_2)
									\end{align*}
									for all \(\tm\in\N^2\). Then, the sequence \((\rpp_{n,\tm})_{\tm\in\N^2}\) satisfies 
									\begin{align*}	\sup_{m_2\in\N} (\rpp_{n,\tm}) = 	\begin{cases}		1,	& \text{if}~ m_1 \in D(g_n), \\
																																									0, & \text{otherwise},
																																	\end{cases}
									\end{align*}
									for all \(n,m_1\in\N\), and hence
									\begin{align*}	\inf_{m_1\in\N}\sup_{m_2\in\N} (\rpp_{n,\tm}) =	\begin{cases}		1,	& \text{if}~ D(g_n) = \N, \\
																																										0, & \text{otherwise},
																																								\end{cases}
									\end{align*}
									for all \(n\in\N\). Following Lemma \ref{lem:ZhengWeihrauchI}, there exists a Turing machine which computes a mapping 
									\((\rpp_{n,\tm})_{\tm\in\N^2} \mapsto (\rp_{n,m})_{m\in\N}\), such that 
									\((\rp_{n,m})_{m\in\N}\) is a computable sequence of rational numbers that satisfies 
									\begin{align*}	\limsup_{m\to\infty} (\rp_{n,m}) = \inf_{m_1\in\N}\sup_{m_2\in\N} (\rpp_{n,\tm}).
									\end{align*}
									Consider the mapping \(\xi : M \mapsto 1\) and define the computable double sequence \((r_{n,\tm})_{\tm\in\N^2}\) of rational numbers via
									\begin{align*}	r_{n,m_1,m_2} := \big(\rp_{n,m_1}\big)^{m_1} 
									\end{align*}
									for all \(n,m_1,m_2\in\N\). The pair \(((r_{n,\tm})_{\tm\in\N^2},\xi)\) is the standard description
									of a computable sequence \((a_{n,m})_{m\in\N}\)  of computable numbers that satisfies 
									\begin{align*} 	\limsup_{m\to\infty} \sqrt[m]{|a_{n,m}|} = \limsup_{m\to\infty} (\rp_{n,m})
																	=		\begin{cases}		1,	& \text{if}~ D(g_n) = \N, \\
																											0, & \text{otherwise}.
																			\end{cases}.
									\end{align*}
									We now define \(\TM : n \mapsto \TM(n) := (a_{n,m})_{m\in\N}\). Then, by \eqref{eq:actual-bandwidth-limsup-sqrt-an}, 
									the signal \(f(n) := \SG(\TM(n))\) satisfies
									\begin{align*} 	\bw(f) = 	\begin{cases}	1,	& \text{if}~ D(g_n) = \N, \\
																				0, 	& \text{otherwise}.
																						\end{cases}
									\end{align*}
									Consequently, for all \(n\in\N\), we have \(\sfrac{1}{2} < \bw(f)\) if and only if \(\en_n\) is a total function.
									Hence, \(\Ob(\SG(\TM(n))) = \OM_{\mathrm{T}}(n)\) holds true for all \(n\in\N\).
	\end{proof}

	\begin{theorem}	There exist a \(f \in \CEs_\pi\) as well as a Turing machine \(\TM\) which, for \(n\in\N\) and \(x\in \Q\), 
									computes a mapping \((x,n) \mapsto m\in\{0,1\}\), such that
									\begin{align*}	\TM(\Oa(f, n), n) = \OM_{\mathrm{T}}(n)  
									\end{align*}
									holds true for all \(n\in\N\). 
	\end{theorem}\begin{proof}		
									We have \(x[\At] \in \Pi_2\), c.f. Lemma \ref{lem:totality_Set_in_Pi2}. Hence, by Theorem \ref{thm:CEs_Pi2},
									there exists a signal \(f\in\CEs_\pi\) such that \(\bw(f) = x[\At]\) holds true. Observe that
									for all \(n\in\N\), we have \(D(\en_n) = \N\) if and only if
									\begin{align*} n \in A_n[x[\At]] = A_n[\bw(f)] 
									\end{align*} 
									is satisfied. In particular, \(D(\en_n) = \N\) holds true if and only if the \(n\)th binary place of 
									\(4\cdot x[A_n[\bw(f)]]\) equals one. Hence, with \(\Oa(f, n) =  x[A_n[\bw(f)]]\), we have \(\bw(f) = x[\At]\) if and only if
									\begin{align*} \underbrace{\Oa(f, n) \cdot 2^{n+2} - \lfloor \Oa(f, n) \cdot 2^{n+1} \rfloor \cdot 2}_{ =: g(\Oa(f, n), n)} = 1.	
									\end{align*}
									Defining
									\begin{align*} 	\TM(\Oa(f, n), n) := 	\begin{cases}	1, &\text{if}~ g(\Oa(f, n), n) = 1, \\
																																			0, &\text{otherwise},
																												\end{cases}
									\end{align*}
									yields the required Turing machine.
	\end{proof}
	
	The remainder of this section will be dedicated to relating the set \(\CBs_\pi^{\infty}\) to 
	the family of halting oracles \((\OM_{\mathrm{H},n})_{n\in\N}\) as \emph{halting oracles}.

	\begin{theorem}
		Let \(f\) satisfy \(f \in \CBs_\pi^{\infty}\). There exists a Turing machine \(\TM\) that computes a mapping
		\begin{align*} \big((\OM_{\mathrm{H},n})_{n\in\N}, f, m\big) \mapsto q \in \Q,
		\end{align*}
		such that \(| \bw(f) - q | < 2^{-m}\) holds true for all \(m \in \N\).
	\end{theorem}\begin{proof}
		The proof again employs the dyadic expansion of the number \(\bw(f)\). Without loss of generality
		we assume \(\bw(f) \in [0,1]\cap\Sigma_1\). There exists a Turing machine \(\TM'\) that computes a mapping
		\(f \mapsto (r_l)_{l\in\N}\), such that \((r_l)_{l\in\N}\) is a monotonically non-decreasing,
		computable sequence of rational numbers, such that \(\lim_{l\to\infty} r_l = \bw(f)\) holds true, c.f. \cite{boche20d} for details.
		Denote by \(A_l\) the largest subset of \(\N\) such that \(x[A_l] \leq \bw(f)\) holds true and define
		\begin{align*}
			A := \bigcup_{l\in\N} A_l.
		\end{align*} 
		Then, we have \(\bw(f) = x[A]\) as well as \(A_l \subseteq A_{l+1}\) for all \(l\in\N\). Furthermore,
		the set \(A\) is recursively enumerable and the mapping \((r_l)_{l\in\N} \mapsto g\), where \(g\) is a recursive
		function that satisfies \(D(g) = A\), can be computed by a Turing machine. Using the family \((\OM_{\mathrm{H},n})_{n\in\N}\)
		of halting oracles, we can now compute the set
		\begin{align*} B_m := \{k \in \N : k \leq m+1, k \in D(g)\}.
		\end{align*}
		For all \(m\in \N\), we have
		\begin{align*}
			B_m = \{k\in A: k \leq m + 1\}.
		\end{align*}
		Since \(A\) contains the dyadic expansion of \(\bw(f)\), we have \(| \bw(f) - x[B_m] | < 2^{-m}\). 
	\end{proof}
	
	\begin{remark} For \(f \in \CEs_\pi\) such that \(\bw(f) \in \Delta_2\) holds true, the family \((\OM_{\mathrm{H},n})_{n\in\N}\)
		allows for the computation of \(\bw(f)\) as well, c.f. \cite[Lemma~5.4, p.~58]{zw01}. However, the dependency on \(f\) is
		non recursive, since the description \((a_n)_{n\in\N}\) in the sense of Definition \ref{def:CEs_pi} is not a feasible input.
	\end{remark}

\section{Conclusion}\label{sec:relation-prior-work}
	In the previous chapters, we have considered different descriptions for the signal classes \(\CBs_\pi^p\)
	for \(1 \leq p \leq +\infty\) and \(\CEs_\pi\), and showed that \(\bw(\CBs_\pi^p)= \Sigma_1\cap[0,\pi]\) and \(\bw(\CEs_\pi)\cap[0,\pi]\) hold true
	for all \(p\) that satisfy \(1 \leq p \leq +\infty\), which is a full characterization of \(\bw(\CBs_\pi^p)\) and \(\bw(\CEs_\pi)\) in terms of the 
	arithmetical hierarchy of real numbers. In the scope of our analysis, we were able to confirm all conjectures posed in \cite{boche20d}: 
	\begin{enumerate}
		\item There exist signals \(f\in \CEs_\pi\) such that \emph{no} computable sequence \((r_n)_{n\in\N}\) of rational numbers
			satisfies \(\lim_{n\to\infty} r_n = \bw(f)\).
		\item For signals \(f\in\CEs_\pi\) and numbers \(\sigma \in (0,\pi)\cap \Rc\), the set \(\{f \in \CEs_\pi : \bw(f) > \sigma\}\)
			is not semi-decidable.
		\item If \(\TMl\) and \(\TMu\) are Turing machines such that for all signals \(f\) in \(\CEs_\pi\) that satisfy \(\bw(f) \leq \pi\), we have
			\begin{align*} \TMl(f) \leq \bw(f) \leq \TMu(f),
			\end{align*}
			then \(\TMl\) and \(\TMu\) necessarily return trivial values, i.e., \(\TMl(f) = 0\) and \(\TMu(f) = \pi\) holds true for all feasible signals \(f\).
	\end{enumerate}
	Last but not least, we showed that the problem of computing the bandwidth of a signal \(f\in\CEs_\pi\) is at least as hard as computing a totality oracle
	for the set of recursive functions. Computing the bandwidth of a signal \(f\in\CBs^p_\pi\), for \(1\leq p\leq +\infty\), on the other hand, can be reduced to
	computing a halting oracle.
	
	It is interesting to note that some of the questions can be positively answered for
	\(\CBs_{\pi}^2\). For example, the set \(\left\{ f \in \CBs_{\pi}^2 \colon \bw(f) > \sigma \right\}\) is semi-decidable.
	Further, for \(f \in \CBs_{\pi}^2\), we can find a computable monotonically increasing sequence of lower bounds for
	\(\bw(f)\) that converges to \(\bw(f)\). The restrictions that we have in the time and frequency domain for signals in
	\(\CBs_{\pi}^1\) and \(\CBs_{\pi}^2\) make it possible that important signal processing
	problems, such as those above, can be algorithmically solved.
	In contrast, signals in \(\CEs_{\pi}\) have---except for a simple growth condition---no such
	restrictions. 

	Questions of computability have not received much attention in control theory and signal
	processing so far.
	Recent results have shown that there are important problems that cannot
	be solved on a digital computer.
	For example, problems can occur in the computation of the Fourier transform
	\cite{boche19g}, the Fourier series \cite{boche19d}, and the spectral factorization
	\cite{boche20a}, as well as in downsampling and the computation of the bandlimited
	interpolation \cite{boche19k}.

	In \cite{boche20c}, conditions were analyzed under which the computability of a
	discrete-time signal implies the computability of the corresponding continuous-time
	signal.
	The computability of the actual bandwidth was studied in \cite{boche20d} and
	\cite{boche21d_accepted}.
	In \cite{boche20d} it was shown that there exist signals \(f \in \CBs_{\pi}^2\) for which
	the actual bandwidth \(\bw(f)\) is not computable.
	In contrast to \cite{boche20d}, we study the most general class of bandlimited signals in
	this paper and the question whether upper and lower bounds can be algorithmically
	determined.

	The problem of computing the period of a periodic computable continuous function is
	similar to the problem studied in this paper. 
	In Shor's famous algorithms for factorizing natural numbers and computing the discrete
	logarithm \cite{shor97}, the core task---and the only part that has to
	be implemented on a universal quantum computer---is the computation of the period of
	certain computable continuous functions.
	Interestingly, all candidates for a ``post-quantum cryptography'' that already failed,
	were broken by quantum algorithms that compute the period of certain functions.
	It seems as if finding periods of functions is the only class of well-investigated
	mathematical problems for which it was possible to develop quantum algorithms that have a
	substantial complexity advantage over the best known classical algorithms.

	As indicated in the introduction, our results contribute to the problem of understanding
	the differences between different approaches for computation.
	Since the problem of sampling analog signals inherently involves analog technology, the
	question arises if, for example, an analog or neuromorphic computer could solve the
	bandwidth estimation problem up to any accuracy. 
	To answer questions of this kind, it is important to understand the mathematical structure
	behind the different approaches for computation. 
	The authors believe that this topic will gain further importance in the future.


	\bibliographystyle{elsarticle-num-names}
	\bibliography{noauthorhyphen,IEEEfull,names_long,literature_part}
\end{document}